\documentclass[11pt]{article}
\usepackage{amsfonts}
\usepackage{graphicx}
\usepackage{grffile}
\input epsf
\usepackage[nosort]{cite}
\usepackage{amssymb,amsmath,amsthm}
\usepackage{enumerate}
\usepackage{tensor}
\usepackage{mathrsfs}
\usepackage{comment}
\usepackage{breqn}
\usepackage{xcolor}
\usepackage{mathtools}
\usepackage{booktabs}
\usepackage{array}

\DeclarePairedDelimiter\floor{\lfloor}{\rfloor}
\usepackage[english]{babel}
\usepackage{braket}
\usepackage{dcolumn}
\usepackage{bm}
\usepackage[hidelinks]{hyperref}
\usepackage{slashed}  
\usepackage{appendix}
\usepackage{algorithm}
\usepackage{algpseudocode}
\makeatletter\@addtoreset{equation}{section}\makeatother

\def\bra#1{\mathinner{\langle{#1}|}}
\def\ket#1{\mathinner{|{#1}\rangle}}

\def\on {\otimes n}

\renewcommand{\title}[1]{\vbox{\center\LARGE{#1}}\vspace{5mm}}
\renewcommand{\author}[1]{\vbox{\center#1}\vspace{5mm}}
\newcommand{\address}[1]{\vbox{\center\em#1}}

\newcommand{\wt}{{\rm{wt}} }

\newcommand{\tr}{{\rm{tr~}}}

\begin{document}
\newtheorem{theorem}{Theorem}
\newtheorem{lemma}{Lemma}
\newtheorem{conjecture}{Conjecture}
\newtheorem{example}{Example}
\newtheorem{corollary}{Corollary}

\begin{titlepage}
\begin{center}

\vskip 1cm

\setcounter{tocdepth}{2}

\title{Invariant Theory, Magic State Distillation, and Bounds on Classical Codes}

\author{Amolak Ratan Kalra$^{1,2}$ and Shiroman Prakash$^{3}$}

\address{
${}^1$ Institute for Quantum Computing, and David R. Cheriton School of Computer Science
University of Waterloo, Waterloo, Ontario, Canada\\
${}^2$ Perimeter Institute for Theoretical Physics
Waterloo, Ontario, Canada\\
${}^{3}$Department of Physics and Computer Science, Dayalbagh Educational Institute, Agra, India}

\end{center}

\begin{abstract}
We show that the physical consistency of magic state distillation imposes new constraints on the weight enumerators of classical error-correcting codes. We establish that for $|T\rangle$-state distillation protocols based on linear self-orthogonal $GF(4)$ codes, the distillation threshold and noise-suppression exponent are directly determined by the code's simple weight enumerator. By enforcing the physical consistency of the distillation process --specifically, that the probability of successfully projecting onto the target state must be non-negative -- we derive a new set of constraints on classical weight enumerators. These ``quantum consistency'' constraints prove to be strictly stronger than those derived from classical invariant theory, yielding new upper bounds on the minimum distance of certain classical and quantum codes. Most notably, we show that these new constraints resolve a long-standing open problem in classical coding theory by proving the non-existence of extremal Hermitian self-dual codes over $GF(4)$ with parameters $[12m, 6m, 4m+2]$. Additionally, we use our formalism to perform an exhaustive search of distillation protocols based on linear $GF(4)$ codes with $n < 20$, finding no protocols with thresholds exceeding the 5-qubit code, and we derive analytic upper bounds on the noise-suppression exponents of such distillation routines as a function of $n$.
\end{abstract}

\vfill

\end{titlepage}

\eject \tableofcontents

\clearpage

\section{Introduction}
\label{Intro}

Magic state distillation \cite{MSD, knill2004faulttolerant} is a protocol for fault-tolerant quantum computing that utilizes quantum error-correcting codes in a manner with no obvious classical analogue. While there have been many advances in our understanding in the decades since \cite{MSD} was first published, (see, e.g., \cite{reichardt2005quantum, Reichardt_2009, ACB, CampbellAnwarBrowne, campbellEnhanced, PhysRevA.86.052329, PhysRevLett.120.050504, Haah2018codesprotocols, 2020golay, prakash2024low, nguyen2024quantum,wills2024constant, golowich2024asymptotically, tamiya2024polylog}), the theoretical foundations of magic state distillation remain poorly understood. Determining the best-attainable noise threshold for magic state distillation is of deep theoretical significance. As argued by Bravyi and Kitaev \cite{MSD}, if any state outside the stabilizer octahedron is distillable, then non-stabilizerness is sufficient for universality; conversely, a stricter bound would imply the existence of a new complexity class intermediate between stabilizer mechanics and universal quantum computation \cite{aaronson2004improved}.

Bravyi and Kitaev \cite{MSD} defined two single-qubit magic states; following their notation, we denotes these as $\ket{H}$ and $\ket{T}$. 
Most work on magic distillation focuses on $\ket{H}$-type magic states, as these can be distilled with very low overhead -- however, as explained by Bravyi and Kitaev \cite{MSD}, to address the theoretical question of the best-attainable threshold of magic state distillation, one must focus on $\ket{T}$ state-distillation. To date, the best-known threshold for distillation of $\ket{T}$ states is achieved by the 5-qubit code originally proposed by Bravyi and Kitaev \cite{MSD} 20 years ago, which is significantly lower than the theoretical upper bound. Moreover, while much is known about families of codes that distill the $\ket{H}$ state, e.g., triorthogonal codes, comparatively little is known about the structure of codes that distill the $\ket{T}$ state. We therefore focus on characterizing codes for  distillation of $\ket{T}$ states in this paper. 

While a generic stabilizer code corresponds to an additive self-orthogonal code over $GF(4)$, Bravyi and Kitaev \cite{MSD} showed that stabilizer codes that correspond to \textit{linear} self-orthogonal codes over $GF(4)$ are naturally suited to distillation of $\ket{T}$ states. This is because $\ket{T}$ states are eigenstates of a order-3 Clifford unitary, known as $M_3$ \cite{rall2017signed}, and stabilizer codes that possess a transversal $M_3$ gate correspond to linear codes over $GF(4)$. Such stabilizer codes are also known as $M_3$-codes \cite{rall2017signed}, and we review their properties in detail in section \ref{sec:basics}.\footnote{The family of $M_3$-codes is quite large, and, is of interest beyond magic state distillation --  in particular, \textit{any} stabilizer code with a complete set of transversal Clifford gates (e.g., the 7-qubit Steane code) is an $M_3$-code.}  Building on the formalism of signed weight-enumerators developed by Rall \cite{rall2017signed}, we establish a rigorous framework connecting the performance of an $M_3$-code for magic state distillation to the code's simple weight enumerator. This mapping allows us to import powerful tools from the theory of polynomial invariants (e.g., \cite{MALLOWS197568,rains, rains2002self}) 
to characterize magic state distillation routines, previously used to constrain the distances of classical and quantum codes.

Furthermore this investigation reveals a surprising inversion of a standard paradigm. Historically, the interaction between classical and quantum error correction has been largely unidirectional. Classical codes are routinely employed as a resource to construct quantum error-correcting codes, often via the Calderbank-Shor-Steane (CSS) construction \cite{calderbank1996good, steane1996multiple}, and more generally via the correspondence between stabilizer codes and additive codes over $GF(4)$ \cite{calderbank1998quantum}.  In this work, we reverse this logic. We find that physical constraints arising from quantum mechanical consistency of distillation via $M_3$-codes imposes new bounds on classical linear self-dual and self-orthogonal codes over $GF(4)$, independent of all previously known constraints. 

These new ``quantum constraints'' resolve a long-standing mystery in classical coding theory, concerning the classification of \textit{extremal self-dual codes}. Extremal self-dual codes are amongst the most symmetric objects in mathematics and play an important role in lattices and sphere packings, the classification of finite groups, combinatorics, and modular forms (see for example \cite{elkies2000lattices,ebeling1994lattices} for more details in this direction.). Of particular interest to us are Hermitian self-dual codes over $GF(4)$, commonly known as Type $4^H$ self-dual codes, which are bounded by the Mallows-Sloane bound \cite{MALLOWS1973188} relating the minimum distance $d$ to the length $n$ via $d \leq 2 \lfloor n/6 \rfloor + 2$. Codes achieving this bound are termed \textit{extremal}. While extremal type IV codes have  been constructed  for many values of $n$, the existence of the family of extremal Type IV codes with parameters $[12m,6m,4m+2]_{GF(4)}$ has remained an open question \cite{macwilliams1978self, lam1990there, huffman1990extremal, huffman1991extremal, rains2002self}. 
Classical linear programming bounds allow, and strongly suggest, the existence of such extremal codes; computer-assisted searches failed to find codes with parameters $[12,6,6]$ or $[24,12,10]$. Our mapping explains these absences and moreover proves the non-existence of any code in the extremal family with parameters $[12m,6m,4m+2]$. We achieve this by showing that any such code would correspond to an $[[n,0]]$ $M_3$-code, i.e. a $12m$-qubit stabilizer state, with negative projection probabilities when applied to $12m$ copies of the magic $\ket{T}$-state. Consequently, the physical requirement of quantum consistency proves the  non-existence of this entire family of classical self-dual codes.

Beyond resolving this open problem, we utilize our framework to explore the landscape of physically realizable distillation protocols. We perform an exhaustive search of all $[[n,1]]$, $M_3$-stabilizer codes with $n<20$, finding no protocols with thresholds exceeding that of the 5-qubit code. We also use our formalism to analytically constrain magic state distillation protocols based on $M_3$-codes. While error-correcting codes are conventionally characterized by distance, magic state distillation routines should instead be characterized by a ``noise-suppression exponent'' that is not equal to the distance. 
For example, the 5-qubit code distills $\ket{T}$ states with quadratic noise suppression although it is a distance $3$ code, demonstrating that distance is not the relevant metric for $\ket{T}$-state distillation. Linear programming and invariant theory has previously been used to bound the distance of quantum error-correcting codes \cite{rains,rains1998quantum}; however, no constraints on noise-suppression of magic state distillation routines have previously appeared in the literature.  We derive upper bounds on the noise-suppression exponents for distillation via $M_3$-codes. We are also able to place stronger upper bounds on the distances of $[[n,1]]$ $M_3$-codes, using our quantum consistency conditions. 

The paper is organized as follows. In Section \ref{sec:basics}, we review magic state distillation, $M_3$-codes and signed weight enumerators. The remaining sections contain new results. In Section \ref{sec:signed-from-simple} we show that the performance of an $M_3$-code for magic state distillation is strictly captured by its simple weight enumerator, and use this result to carry out an exhaustive search for distillation routines over all $[[n,1]]$ $M_3$-codes with $n \leq 20$. In Section \ref{sec:constraints}, after reviewing classical constraints, we derive new quantum  constraints on weight enumerators that arise from demanding non-negative success probabilities and thresholds outside the stabilizer polytope; we then show these place much stronger bounds on weight-enumerators than classical constraints alone. In Section \ref{sec:self-dual}, we apply these constraints to rule out the $[12m, 6m, 4m+2]$ extremal type $4^H$ self-dual codes .  In Section \ref{sec:distillation-routines}, we use our formalism to place bounds on the best-attainable noise suppression exponent of a magic state distillation routine based on an $[[n,1]]$ $M_3$-code. We conclude with a discussion in Section \ref{discussion}. In Appendix \ref{app:small-examples}, we explicitly compute the space of weight enumerators for $[[n,0]]$ and $[[n,1]]$ $M_3$-codes consistent with classical and quantum constraints, and compare them to weight enumerators of actual codes. In Appendix \ref{app:integer-programming}, we provide integral weight enumerators for putative codes with lengths $23 \leq n \leq 35$ that would exhibit high thresholds, though their existence remains an open question.

\section{Magic state distillation and signed weight enumerators}
\label{sec:basics}

\subsection{Review of magic state distillation}

Here we provide a very brief review of magic state distillation, primarily for the purpose of establishing notation and conventions. The reader is encouraged to consult \cite{MSD} for more details.

In the magic state model, we assume access to a quantum computer restricted to \textbf{fault-tolerant stabilizer operations}. Specifically, the device can:
\begin{itemize}
    \item Prepare qubits in the computational basis states $\ket{0}$ and $\ket{1}$.
    \item Apply unitary gates from the Clifford group.
    \item Perform measurements of Pauli observables.
\end{itemize}
These operations are possible to implement fault-tolerantly, using, e.g., the surface code or Ising anyons. However, these operations can also be efficiently classically simulated, via the Gottesman-Knill theorem\cite{aaronson2004improved}.
To achieve universality, Bravyi and Kitaev proposed that we supplement fault-tolerant stabilizer operations with the ability to prepare ancillae in certain non-stabilizer states known as magic states. We focus on the magic state $\ket{T}$, given by,
\begin{equation}    \ket{T}\bra{T} =  \frac{1}{2}\left( I+ \frac{1}{\sqrt{3}} (X+Y+Z) \right).
\end{equation}Bravyi and Kitaev showed that, using $O(1)$ pure $\ket{T}$ states, one can implement a non-Clifford gate via state injection, thereby achieving universality. 

However, our quantum computer cannot produce $\ket{T}$ states fault-tolerantly; any magic states produced will be noisy. Crucially, Bravyi and Kitaev showed that it is possible to distill one $\ket{T}$ state of arbitrarily high-fidelity from many low-fidelity $\ket{T}$ states, provided that the fidelity of the input states exceeds a critical threshold, via a process only involving Clifford unitaries and Pauli measurements, known as magic state distillation. 

\begin{figure}
    \centering
    \includegraphics[width=0.6\linewidth]{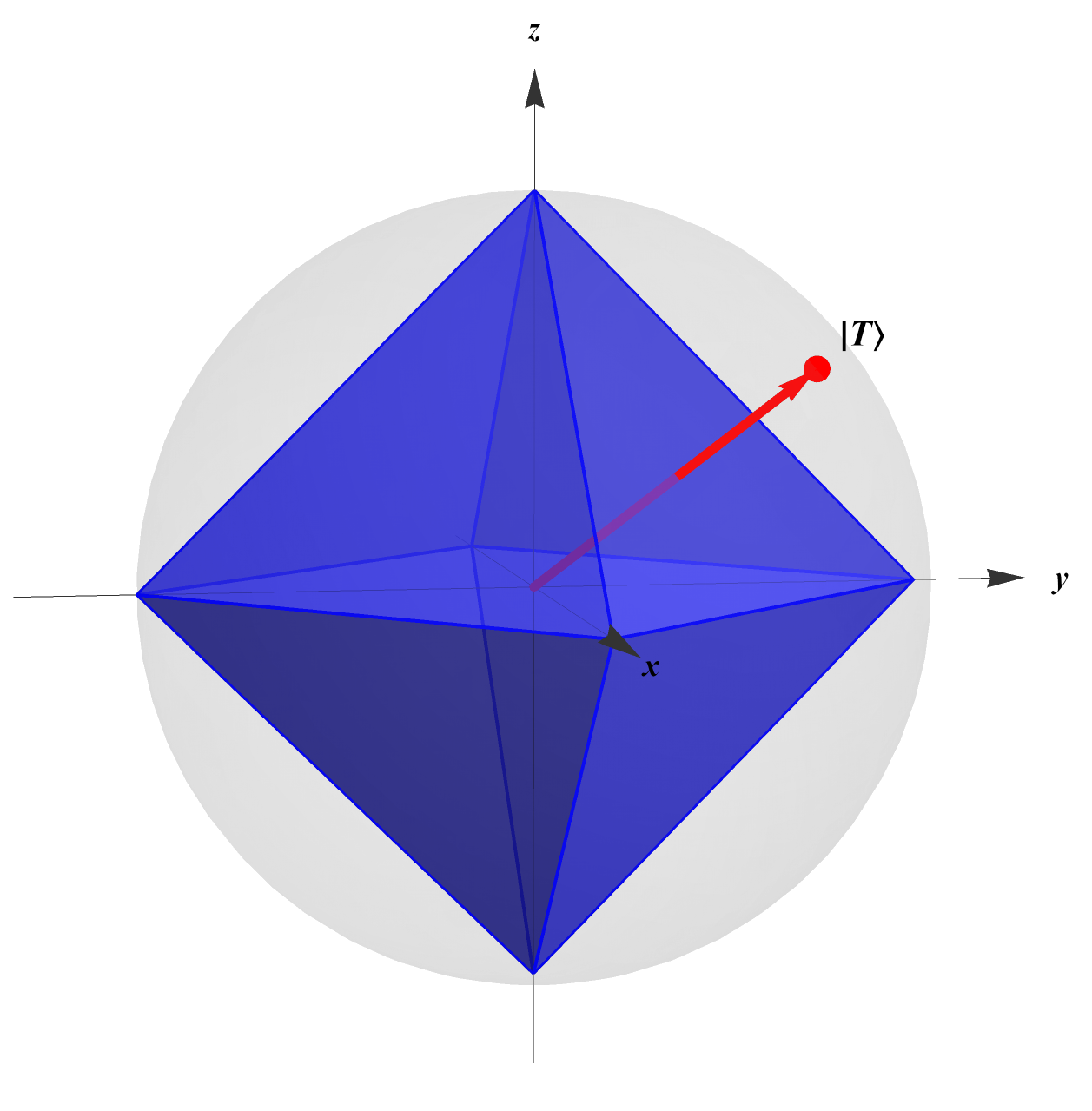}
    \caption{The stabilizer octahedron and the  $\ket{T}$ state.}
    \label{fig:octahedron}
\end{figure}
Classical-simulability and magic state distillation are closely related. Define the stabilizer octahedron to be the set of convex mixtures of single-qubit stabilizer states, shown in Figure \ref{fig:octahedron}. Mixed states within the stabilizer octahedron cannot be used to achieve universal quantum computation, as Clifford operations and stabilizer measurements on such states are classically simulable via the Gottesman-Knill theorem \cite{aaronson2004improved}. If a quantum computer is able to produce noisy magic states that lie outside the stabilizer octahedron, it cannot be efficiently simulated via the Gottesman-Knill theorem. However, is the ability to produce noisy magic states outside the stabilizer octahedron enough to achieve quantum universality? If not, this implies a new complexity class exists between stabilizer mechanics and universal quantum computation. To demonstrate that that non-stabilizerness is not only necessary but also sufficient for universal quantum computing, one would have to show that, given a sufficiently large supply of noisy magic states outside the stabilizer octahedron, they can be distilled to an arbitrarily high-fidelity magic state. 

\subsection{Magic state distillation and $M_3$-codes}

It was shown by Campbell and Browne  \cite{campbell2009structure} that any magic state distillation protocol can be presented as a stabilizer reduction -- i.e., projection of $n$ noisy magic states onto the codespace of an $[[n,k]]$ stabilizer code, followed by decoding. By demanding the stabilizer code used for distillation possesses certain symmetry, we can simplify the search for, and analysis of, magic state distillation protocols considerably.

The $\ket{T}$ state is an eigenvector of an order-3 element of the Clifford group, 
\begin{equation}
M_{3}=\frac{e^{\frac{i\pi}{4}}}{\sqrt{2}}\begin{pmatrix}
1 & 1\\
i & -i\\
\end{pmatrix},
\end{equation} which acts on Pauli operators as follows,
\begin{equation}
    M_3^\dagger X M_3 = Y, \quad M_3^\dagger Y M_3 = Z, \quad M_3^\dagger Z M_3 = X.
\end{equation}

As depicted in Figure \ref{fig:octahedron}, the $\ket{T}$ state lies directly above a face of the stabilizer octahedron. On the Bloch sphere, acting with $M_3$ corresponds to a $\frac{2\pi}{3}$-rotation around the axis connecting the maximally mixed state to $\ket{T}$, depicted as a red line in Figure \ref{fig:octahedron}. 

We define an \textbf{$\mathbf{M}_3$-code} as an $[[n,k]]$ stabilizer code, whose projector, $\Pi_S$, commutes with $M_3^{\otimes n}$: \begin{equation}
    [M_3^{\otimes n}, \Pi_S]=0. \label{M3-condition}
\end{equation} 
Because they respect the Clifford symmetry of the $\ket{T}$ state, $M_3$-codes are natural candidates for magic state distillation\footnote{We should remark that, while all known protocols for distilling the $\ket{T}$ state are based on $M_3$-codes, there is no proof that these are the most general protocols for $\ket{T}$ state distillation.}. 
 For distillation, we focus on $[[n,1]]$ $M_3$-codes. 
For such codes the logical operators can always be chosen such that the $M_3^{\otimes n}$ implements a logical $\overline{M}_3$ or $\overline{M}_3^\dagger$ gate. 
The standard 5-qubit code \cite{MSD} is the simplest example of such a code.

Given an arbitrary single-qubit mixed state \begin{equation}
\rho(a_X, a_Y, a_Z)= \frac{1}{2}\left( I + a_X X+ a_Y Y+a_Z Z) \right),
\end{equation} we can produce a state $\rho_{T}$ via the \textit{twirling} operation
\begin{equation}
    \hat \rho \to \rho_{T}= \frac{1}{3} \hat \rho + \frac{1}{3} M_3 \hat \rho M_3^\dagger + \frac{1}{3} M_3^2 \hat \rho (M_3^{\dagger})^2.
\end{equation} 
Consider any noisy magic state $\rho(a_X,a_Y,a_Z)$. First apply Clifford unitaries to ensure that $a_X\geq0$, $a_Y\geq0$ and $a_Z\geq0$, then twirl the magic state. After twirling, the noisy magic state will take the form
\begin{equation}
\rho_{T}(r)= \frac{1}{2}\left( I + \frac{r}{\sqrt{3}}(X+Y+Z) \right). \label{twirled-state}
\end{equation}
The twirled magic state depends on a single parameter $r$, which represents the radial distance from the center of the Bloch sphere (along the red line in Figure \ref{fig:octahedron}). In terms of the input coefficients, $r$ is given by,
\begin{equation}
    r=\frac{1}{\sqrt{3}}(a_X+a_Y+a_Z).
\end{equation}
$\rho_{\rm twirled}$ is outside the stabilizer octahedron if and only if the untwirled noisy magic state $\rho(a_X,a_Y,a_Z)$ is outside the stabilizer octahedron. It is also convenient to introduce the error-rate $\epsilon$, which is related to $r$ via $r=1-2\epsilon$. 

For distillation via an $[[n,1]]$ $M_3$-stabilizer code (with logical operators chosen so that $M_3^{\otimes n}$ acts as $\overline{M}_3$ or $\overline{M}_3^{-1}$), the output of the distillation routine is also of the form given in Equation \eqref{twirled-state}. The relation between input and output of a magic state distillation routine can then be captured by a single function $\epsilon_{\rm out} = f_{\rm MSD}(\epsilon).$ For the 5-qubit code, this function was computed by Bravyi and Kitaev \cite{MSD} to be,
\begin{equation}
    \epsilon_{\rm out} = \frac{\epsilon_{\rm in} ^2 \left(5-15 \epsilon_{\rm in} +15 \epsilon_{\rm in} ^2-4 \epsilon_{\rm in}   ^3\right)}{1-5 \epsilon_{\rm in} +15 \epsilon_{\rm in} ^2-20 \epsilon_{\rm in} ^3+10 \epsilon_{\rm in} ^4}. \label{msd-5}
\end{equation}
\begin{figure}
    \centering
    \includegraphics[width=0.6\linewidth]{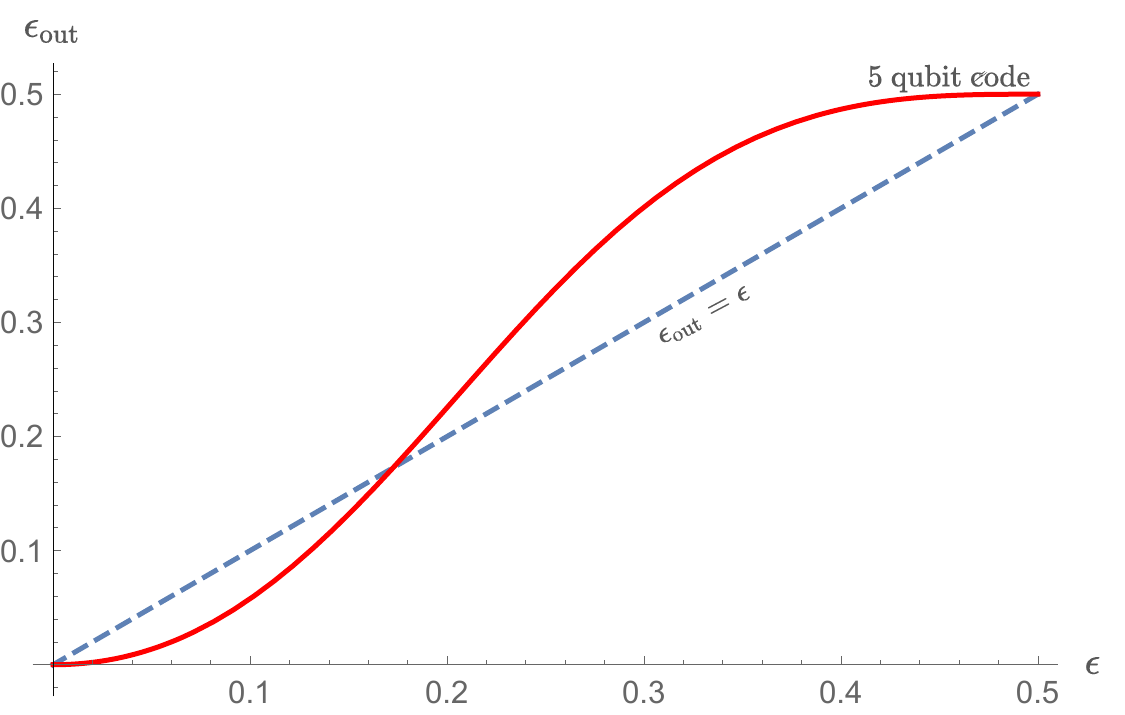}
    \caption{The distillation performance, $\epsilon_{\rm out}=f_{\rm MSD}(\epsilon_{\rm in})$, of the 5-qubit code is plotted in red. The intersection of this curve with the  line $\epsilon_{\rm out}=\epsilon_{\rm in}$, shown as a blue dashed line, determines the threshold of the 5-qubit code, which is $\approx 0.17$.}
    \label{fig:5-qubit-distillation}
\end{figure}
Equation \eqref{msd-5} is plotted in Figure \ref{fig:5-qubit-distillation}. The \textbf{threshold} of a magic state distillation routine is a fixed point $\epsilon_*$ of $f_{\rm MSD}$ such that $\epsilon_{\rm out}<\epsilon_{\rm in}$ for all $\epsilon_{\rm in} < \epsilon_*$. For the five qubit code, we obtain the threshold,
\begin{equation}
\epsilon_*^{\rm (5-qubit)} = \frac{1}{14} \left(7-\sqrt{21}\right) \approx 0.17.
\end{equation}
To date, this is the best known threshold of any qubit magic state distillation routine. 

The theoretical upper-bound on the threshold set by the classical simulability is $\epsilon_{\rm max}=\frac{1}{2}\left(1-\frac{1}{\sqrt{3}}\right)\approx 0.21$ which is where the face of the stabilizer octahedron intersects the red line in Figure \ref{fig:octahedron}. It can be shown that no $[[n,1]]$ distillation protocol can achieve this bound \cite{campbell2010bound}.  A stabilizer quantum computer with access to a supply of noisy $\ket{T}$-states with $\epsilon \in (0.17, \epsilon_{\rm max})$ is not classically simulable. However, at present, it is not known how to use such a quantum computer to achieve universal quantum computation. If a distillation protocol could be found with a threshold better than $0.17$, it would reduce the size of this non-simulable-and-possibly non-universal region of state space. And, if a family of protocols can be found whose threshold approaches $\epsilon_{\rm max}$, it would demonstrate that non-stabilizerness is sufficient for quantum computation. 

Let us conclude with some terminology that we will use in the following sections. We refer to the probability of successfully projecting $n$ twirled noisy magic states onto a stabilizer code as the \textbf{success probability} $\eta(\epsilon)$. We refer to the function $\epsilon_{\rm out}=f_{\rm MSD}(\epsilon_{\rm in})$ as the \textbf{distillation performance} of a stabilizer code. For small $\epsilon_{\rm in}$, we expect $f_{\rm MSD} \sim A \epsilon_{\rm in}^\nu$ for some $\nu$ -- we refer to $\nu$ as the \textbf{noise suppression exponent} of the distillation protocol. $\nu$ is usually not equal to the distance of the stabilizer code -- for example, although the $5$-qubit code is a distance $3$ quantum error-correcting code, we see from equation \eqref{msd-5} that  $\epsilon_{\rm out} \sim 5 \epsilon_{\rm in}^2$, i.e., it has $\nu=2$.

\subsection{Signed weight enumerators}
In this subsection, we review the work of Rall \cite{rall2017signed}, who introduced signed weight enumerators to characterize the performance of qubit magic state distillation routines.  

\subsubsection*{Stabilizer projectors and signed weight enumerators}
\label{sec:stabilizer-projection}

Here, we will arrive at the definition of signed weight enumerators by computing the probability for successfully projecting $n$ copies of a single-qubit mixed state onto the codespace of a stabilizer code. We use the following notation. Let $I$ be the identity operator, and $X$, $Y$ and $Z$ be the Pauli matrices. Let $\mathcal P_n$ be the set of all $n$-fold tensor products of $I$, $\pm X$, $\pm Y$, $\pm Z$ and let $\mathcal P_n^+$ denote the set of all $n$-fold tensor products of $I$, $+X$, $+Y$, $+Z$. We refer to $\mathcal P_n^+$ as the set of positive Pauli operators. Any element $\tilde{P}\in \mathcal P_n$ can be written in the form $\tilde{P}= \lambda P$ with $\lambda = \pm 1$ and $P \in P_n^+$.

An $[[n,k]]$ stabilizer code is defined as the joint eigenspace of a collection of $n-k$ independent, commuting $n$-qubit Pauli operators. These operators generate an abelian subgroup $\mathcal S$ of $\mathcal P_n$. Each Pauli operator has two eigenvalues $+1$ and $-1$; therefore, to completely specify the eigenspace, one must also specify the eigenvalue for each Pauli operator $P \in \mathcal S$, which we denote as $\lambda(P)=\pm 1$. The projector onto the codespace can be written as 
\begin{equation}
    \Pi_S = \frac{1}{2^{n-k}} \sum_{\lambda(P) P \in \mathcal S} \lambda(P)P,
\end{equation}
where, to avoid ambiguity, we will always specify  each $\tilde{P} \in \mathcal S$ as $\tilde{P}=\lambda(P)P$, i.e., the product of a sign, $\lambda(P)= \pm 1$, and a positive Pauli operator, $P \in \mathcal P^+$.

We write a general single-qubit density matrix as,
\begin{equation}
\rho(\vec{a}) = \frac{1}{2}\left( a_I + a_X X+a_Y Y +a_Z Z\right) = \frac{1}{2}\sum_{Q \in \{I,X,Y,Z\}} a_Q Q,
\end{equation}
where $\vec{a}=(a_I, a_X, a_Y, a_Z)$ and $a_I=1$ if the state is normalized. 
Let us compute the probability 
\begin{equation}
    \eta = \tr [\Pi_{\mathcal S} \rho(\vec{a})^{\otimes n}],
\end{equation} 
of successfully projecting the $n$-fold tensor product of $\rho(\vec{a})$ onto the stabilizer code $\mathcal S$. Let $Q$ be a single-qubit Pauli operator, and let $\wt_Q(P)$ denote the total number of $Q$'s in a multi-qubit Pauli operator $P$, so, e.g., $\wt_Z(XZZXI)=2$. We can write,
\begin{equation}
\rho(\vec{a})^{\otimes n} = \frac{1}{2^n}\sum_{P \in {P_n^+}} a_I^{\wt_I(P)}a_X^{\wt_X(P)}a_Y^{\wt_Y(P)}a_Z^{\wt_Z(P)} P.
\end{equation}
Now,
\begin{equation}
 \tr (P\Pi_{\mathcal S} ) = \begin{cases} 2^k \lambda(P) & \lambda(P) P \in \mathcal S \\ 0 & \lambda(P) P \notin \mathcal S. \end{cases}
\end{equation} Hence,
\begin{equation}
\begin{split}
    \eta = \tr (\Pi_{\mathcal S}\rho(\vec{a})^{\otimes n})  
    & = \frac{1}{2^{n-k}}\sum_{\lambda(P) P \in {\mathcal S}} a_I^{\wt_I(P)}a_X^{\wt_X(P)}a_Y^{\wt_Y(P)}a_Z^{\wt_Z(P)} \lambda(P). \label{stabilizer-projection-signed-cwe}
\end{split}
\end{equation}
This expresses the probability of successful projection onto the stabilizer code, $\eta$, in terms of the \textbf{signed complete weight enumerator} of the stabilizer code $\mathcal S$, defined as,
\begin{equation}
    W_I(a_I,a_X, a_Y, a_Z) = \sum_{\lambda(P) P \in {\mathcal S}} a_I^{\wt_I(P)}a_X^{\wt_X(P)}a_Y^{\wt_Y(P)}a_Z^{\wt_Z(P)} \lambda(P). 
\end{equation}
Note that the $\eta$ must be non-negative for physical density matrices because it is a probability. 

Let $Q_L$ be any $n$-qubit Pauli operator a representative of a logical operator in $N(\mathcal S)/\mathcal S$. Then, we also have,
\begin{equation}
\begin{split}
   \tr (\Pi_{\mathcal{S}} \rho(\vec{a})^{\otimes n}Q_L ) & =  \frac{1}{2^{n-k}} \sum_{\lambda(P) P \in Q_L {\mathcal S}} a_I^{\wt_I(P)}a_X^{\wt_X(P)}a_Y^{\wt_Y(P)}a_Z^{\wt_Z(P)} \lambda(P), \\
   & \equiv 
    \frac{1}{2^{n-k}} W_L(a_I,a_X, a_Y, a_Z), 
\end{split}
\end{equation}
where $Q_L \mathcal S$ denotes the set of all Pauli operators representative of the logical operator $L$ in the stabilizer code. 

In this paper, we will primarily be interested in the states of the form $\rho_{T}(r)^{\otimes n}$, where $\rho_{T}$ is given in equation \eqref{twirled-state}.
For such states, the above expressions reduce to,
\begin{equation}
\begin{split}
    \eta = \tr (\Pi_{\mathcal S}\rho_{T}(r)^{\otimes n}) & = \frac{1}{2^{n-k}} \sum_{\lambda(P) P \in {\mathcal S}} \left(\frac{r}{\sqrt{3}}\right)^{\wt(P)} \lambda(P),
\end{split}
\end{equation}
and 
\begin{equation}
\begin{split}
    \tr (Q_L \Pi_{\mathcal S}\rho_{T}(r)^{\otimes n}) & = \frac{1}{2^{n-k}} \sum_{\lambda(P) P \in {Q_L \mathcal S}} \left(\frac{r}{\sqrt{3}}\right)^{\wt(P)} \lambda(P),
\end{split}
\end{equation}
where $\wt(P)$ denotes the Hamming weight of the Pauli operator $P$.
Using the notation of Rall \cite{rall2017signed}, we therefore define the \textbf{signed simple weight enumerators} of a stabilizer code, 
\begin{eqnarray}
    W_I(\bar{r}) & = & \sum_{\lambda(P) P \in {\mathcal S}} {\bar{r}}^{\wt(P)} \lambda(P), \\
    W_{L}(\bar{r}) & = & \sum_{\lambda(P) P \in {Q_L \mathcal S}} \bar{r}^{\wt(P)} \lambda(P), \label{signed-logical-swe}
\end{eqnarray}
where, for convenience, we define $\bar{r} \equiv r/\sqrt{3}$ following \cite{rall2017signed}.

\subsubsection*{Magic state distillation and signed weight enumerators}

With the above results in place, it is easy to see that signed weight enumerators characterize magic state distillation. For simplicity, we restrict attention to $n$-to-$1$ magic state distillation protocols that takes $n$ twirled noisy $\ket{T}$ states as input. The twirled $n$-qubit state, $\rho^{(n)}_{\rm in}=\rho_T(r_{\rm in})^{\otimes n}$ is purified by projecting it onto the codespace of an $[[n,1]]$ stabilizer code $\mathcal{S}$, and then decoding to obtain the normalized single-qubit density matrix, $\rho_{\rm out}$, which we write as,
\begin{equation}
    \rho_{\rm out}= \frac{1}{2} \left( {a}^{\rm out}_{I} I + {a}^{\rm out}_{X} X+ {a}^{\rm out}_{Y}Y + {a}^{\rm out}_{Z} Z\right) = \frac{1}{2}\sum_{L \in \{I,X,Y,Z\}} {a}^{\rm out}_L L.
\end{equation}
After normalizing by the probability of successful projection, each $a_L^{\rm out}$ is given by
\begin{equation}
    {a}_L^{\rm out} = \frac{W_L(\bar{r}_{\rm in})}{W_I(\bar{r}_{\rm in})}.
\end{equation}
 
For distillation by an $[[n,1]]$ $M_3$-code, we further have $W_X=W_Y=W_Z \equiv W_L$. The output magic state has ${a}^{\rm out}_X={a}^{\rm out}_Y={a}^{\rm out}_Z \equiv \frac{r_{\rm out}}{\sqrt{3}}$, and can be written in the form $\rho_T(r_{\rm out})$. The performance of the magic state distillation protocol is therefore encoded in the function, 
\begin{equation}
    \bar{r}_{\rm out}\equiv r_{\rm out}/\sqrt{3} = W_L(\bar{r})/W_I(\bar{r}). \label{r-msd}
\end{equation}
It is convenient to instead work with the error rate $\epsilon$; using $\bar{r}(\epsilon)=\frac{1-2\epsilon}{\sqrt{3}}$, we find,
\begin{equation}
    \epsilon_{\rm out} = \frac{W_I(\bar r(\epsilon_{\rm in}))-\sqrt{3} W_L(\bar r(\epsilon_{\rm in}))}{2 W_I(\bar r(\epsilon_{\rm in}))}. \label{epsilon-msd}
\end{equation}
\subsection{The structure of $M_3$-codes}
\label{sec:M3-codes}
Before we proceed further, we need to review some technical properties of $M_3$-codes that we will use. These results are based on \cite{macwilliams1978self, calderbank1998quantum, MSD, rall2017signed}.

The connection between stabilizer codes and classical coding theory is established through the mapping of Pauli operators to vectors over the finite field $GF(4) = \{0, 1, \omega, \omega^2\}$, where $\omega^2 = \omega + 1$. Following \cite{calderbank1998quantum}, we map the single-qubit Pauli group (modulo phases) to $GF(4)$ via the isomorphism:
\begin{equation}
    I \to 0, \quad X \to 1, \quad Z \to \omega, \quad Y \to \omega^2.
\end{equation}
An $n$-qubit Pauli operator $P$ maps to a vector $v \in GF(4)^n$. Under this mapping, the commutation relation between two Pauli operators corresponds to the Hermitian inner product of their vectors.

For a general stabilizer code, the stabilizer group maps to an additive code over $GF(4)$ that is self-orthogonal with respect to this inner product. However, the additional symmetry of $M_3$-codes—specifically, that $[M_3^{\otimes n}, \mathcal{S}] = 0$—restricts the structure further. Since conjugation by $M_3$ corresponds to multiplication by $\omega$ in $GF(4)$ (cycling $X \to Y \to Z \to X$ corresponds to $1 \to \omega^2 \to \omega \to 1$), the code must be closed under scalar multiplication by $\omega$. Consequently, $M_3$-codes correspond to \textit{linear} subspaces of $GF(4)^n$, rather than merely additive ones. This implies that the number of stabilizers $n-k$ must be even for any $[[n,k]]$ $M_3$-code. In particular, $[[n,0]]$ $M_3$-codes only exist for even $n$, and correspond to linear classical self-dual codes over $GF(4)$; and $[[n,1]]$ $M_3$-codes only exist for odd $n$, and correspond to maximal self-orthogonal codes over $GF(4)$. 

An example of an $[[n,0]]$ $M_3$-code is the six qubit perfect tensor \cite{pastawski2015holographic}, whose stabilizers, 
\begin{equation}
    \begin{split}
     \{   +ZXXZII, +XYYXII,  +IXYYXI,  +IZXXZI, +IIXYYX, +IIZXXZ\},
    \end{split}
\end{equation}
correspond to a classical linear self-dual $[6,3,4] _{GF(4)}$ code known as the hexacode \cite{macwilliams1978self}. Examples of $[[n,1]]$ $M_3$-codes include the 5-qubit code and the 7-qubit Steane code.

For a classical code $\mathcal{C}$ of length $n$, corresponding to the stabilizers $\mathcal{S}$ of the quantum code, we define the \textbf{simple weight enumerator} \cite{MacWilliamsSloane5} as the homogeneous polynomial:
\begin{equation}
    A(x,y) = \sum_{c \in \mathcal{C}} x^{n-\wt(c)} y^{\wt(c)} = \sum_{j=0}^n A_j x^{n-j} y^j,
\end{equation}
where $A_j$ is the number of codewords of Hamming weight $j$. Similarly, let $B(x,y)$ denote the weight enumerator of the dual code $\mathcal{C}^\perp$, which corresponds to $N(\mathcal S)$.
These are related by the MacWilliams identity \cite{MacWilliamsSloane5}:
\begin{equation}
    B(x,y) = \frac{1}{2^{n-k}} A(x+3y, x-y), \label{eq:classical-macwilliams}
\end{equation}
This polynomial characterizes the distribution of weights in the classical code. We also define
\begin{equation}
    C(x,y)= \sum_{c \in \mathcal C^\perp/\mathcal C} x^{n- \wt(c)}y^{\wt(c)} =B(x,y)-A(x,y) ,
\end{equation}
which is the weight enumerator for logical operators of the stabilizer code, $N(\mathcal S)/\mathcal S$.

For a stabilizer code $\mathcal S$ with the simplest choice of signs $\lambda(P)=+1$ for all $P \in \mathcal S$, the signed weight enumerator $W_I(\bar{r})$ is equal to the unsigned weight enumerator $A(1,r)$. However, more generally, and particularly in the context of magic state distillation, we encounter stabilizer codes $\mathcal S$ that contain some Pauli operators with $\lambda(P)=-1$. For such codes, unsigned classical weight enumerators differ from the signed weight enumerators defined in section \ref{sec:stabilizer-projection}.

Crucially, for $M_3$-codes, the signs of the stabilizer generators are not free parameters; they are fixed by the weights of the Pauli operators by the following theorem, proven by Rall \cite{rall2017signed}, which we refer to as \textit{Rall's rule}.

\begin{theorem}[\textbf{Rall's rule}]
\label{thm:ralls-rule}Let $\mathcal S$ be an $M_3$-code. Then, for any $P \in S$,
\begin{enumerate}
\item ${\rm wt}~P$ is even.
\item If ${\rm wt }~P \equiv 0 \pmod 4$, $\lambda(P)=+1$.
\item If ${\rm wt }~P \equiv 2 \pmod 4$, $\lambda(P)=-1$.
\end{enumerate}
\end{theorem}
\begin{proof} 
For any $\tilde{P}\in \mathcal P_n$ define $\tilde{P}'={M_{3}^{\dagger}}^{\otimes n}\tilde{P}M_{3}^{\otimes n}$ and $\tilde{P}''={M_{3}^{\dagger}}^{\otimes n}\tilde{P}'M_{3}^{\otimes n}$. Consider any $\tilde{P}\in \mathcal S$. Write $\tilde{P}=\lambda(P)P$ for $P \in \mathcal P_n^+$. $\tilde{P}'=\lambda(P)P'$ and $\tilde{P}''=\lambda(P)P''$ must also be in $\mathcal S$ since $\mathcal S$ is an $M_3$-code. Using $XY=iZ$, $YZ=iX$ and $ZX=iY$, it follows that $PP'=i^{\wt(P)}P''$, and $P'P=(-i)^{\wt(P)}P''$.
\begin{enumerate}
    \item  Demanding $[\tilde{P},\tilde{P'}]=0$ implies $\wt(P)$ is even.
    \item Using $\tilde{P}''= \lambda(P)P''$ and $\tilde{P}''=\tilde{P} \tilde{P}'=\lambda(P)^2i^{\wt(P)} P''=i^{\wt(P)} P''$, we see that $\lambda(P)=i^{\wt(P)}$.
\end{enumerate}
\end{proof}
Note also that the total number of stabilizers of any non-zero weight must be divisible by 3. Another useful property of $M_3$-codes we will use is the following lemma. 
\begin{lemma}
\label{logical-operator-weight-n-lemma}
Every $[[n,1]]$ $M_3$-code possesses a logical operator of weight $n$.
\end{lemma}
\begin{proof}
Let $A(x,y) = \sum A_j x^{n-j}y^j$ be the simple weight enumerator of the classical code corresponding to the stabilizer group $\mathcal{S}$. By Rall's rule, all stabilizers have even weight, so $A_j = 0$ for all odd $j$.

The weight enumerator of the dual code, $B(x,y) = \sum B_j x^{n-j}y^j$, determines the weights of operators in the normalizer $N(\mathcal{S})$. Using the MacWilliams identity (Eq. \eqref{eq:classical-macwilliams}), the number of weight-$n$ operators in the dual is given by the coefficient of $y^n$:
\begin{equation}
    B_n = B(0,1) = \frac{1}{2^{n-1}} A(3, -1).
\end{equation}
Expanding the polynomial $A(3, -1)$:
\begin{equation}
    B_n = \frac{1}{2^{n-1}} \sum_{j=0}^n A_j 3^{n-j} (-1)^j.
\end{equation}
Since $A_j = 0$ for all odd $j$, the term $(-1)^j$ is always equal to $+1$ for all non-zero $A_j$. Furthermore, since $A_0=1$ and $A_j \geq 0$, the sum is strictly positive. Thus, $B_n > 0$.

Since $n$ must be odd for any $[[n,1]]$ $M_3$-code, this weight-$n$ operator cannot be a stabilizer (as stabilizers must have even weight). Therefore, it must be a logical operator.
\end{proof}

\section{Magic state distillation and unsigned weight enumerators}
\label{sec:signed-from-simple}
Signed weight enumerators are relatively mysterious compared to their unsigned counterparts which have been extensively studied both in classical and quantum coding theory, see \cite{nebe2006self,rains2002self, shor1997quantum, rains1998quantum,rains1999quantum,rains2000polynomial}. In this section, we will show that for the family of $M_3$-codes, the signed weight enumerator is entirely determined by the unsigned simple weight enumerator. This allows us to characterize the distillation performance of these codes using standard results from classical coding theory.

%For qudits of odd prime dimension, using a formalism based on discrete Wigner functions,  \cite{prakash2024search} related the performance of a stabilizer code for magic state distillation to its unsigned complete weight enumerator. \cite{prakash2024search} further showed that the unsigned simple weight enumerator characterizes the performance of a distillation routine for a particular qutrit magic state known as the strange state \cite{PhysRevA.83.032310,veitch2014resource,jain2020qutrit, 2020golay}.  In order to characterize magic state distillation via unsigned weight enumerators  \cite{prakash2024search} demanded that a particular Clifford gate, the square of the single-qudit Hadamard operator, be exactly transversal for the stabilizer code used for distillation. This condition fixes the eigenvalues of all stabilizers in the code to be $+1$. For qubits, the square of the Hadamard gate is the identity, so we cannot demand its transversality to fix the eigenvalues of all stabilizers in a code. As we will see, demanding transversality of another single-qubit Clifford gate, the $M_3$ gate defined below, determines the eigenvalues of each stabilizer in a code in terms of its Hamming weight, thus allowing one to express the signed weight enumerator in terms of unsigned weight enumerators. We then present the relation to simple weight enumerators, and use the result to carry out a somewhat exhaustive search for $\ket{T}$-state distillation routines.

\subsection{Signed enumerators from unsigned enumerators}
\label{simple-from-signed}
Rall's rule (Theorem \ref{thm:ralls-rule}) implies that for any $M_3$-code, the phase of a stabilizer is determined by $\lambda(P) = i^{\wt(P)}$. This allows us to map the simple weight enumerator $A(x,y)$ of the classical code $\mathcal C$ directly to the signed weight enumerator of the quantum code $\mathcal S$.

\begin{corollary}
Let $\mathcal S$ be an $M_3$-code and $A(x,y)$ be the simple weight enumerator of the corresponding classical $GF(4)$ code. The signed weight enumerator of the $M_3$-code is:
\begin{equation}
    W_I(\bar{r}) = A(1, i\bar{r}).
\end{equation}
\end{corollary}
\noindent This follows immediately by substituting $y \to i\bar{r}$ into the polynomial $A(1,y)$, effectively applying the factor $(i)^{\wt(P)}$ to terms of weight $\wt(P)$. In other words,  the probability of successfully projecting $n$ noisy magic states onto the codespace of an $M_3$-code is determined by its simple weight enumerator. 

We next wish to relate $W_L(\bar{r})$ to the unsigned classical weight enumerator $C(x,y)$. This requires us to determine the sign of each Pauli operator $P \in N(\mathcal S)/\mathcal S$. In order to do this, we must classify $[[n,1]]$ $M_3$-codes into two types, based on the following lemma. 
\begin{lemma}
    Any magic state distillation routine based on an $[[n,1]]$ $M_3$-code is equivalent to one in which $M_3^{\otimes n}$ acts as $\overline{M}_3$ or $\overline{M}_3^{-1}$ on the encoded qubit.  \label{lemma:two-classes}
\end{lemma}
\begin{proof}
Note that two $M_3$-codes are equivalent for magic state distillation if they differ only by local $M_3$ unitaries, as the twirled input state $\rho_T^{\otimes n}$ is invariant under such transformations. Since every $[[n,1]]$ $M_3$-code contains a logical operator of weight $n$ (Lemma \ref{logical-operator-weight-n-lemma}), we can always apply local $M_3$ corrections to map this operator to $X^{\otimes n}$. Consequently, without loss of generality, we assume the logical operators are $\pm X^{\otimes n}$, $\pm Y^{\otimes n}$, and $\pm Z^{\otimes n}$. $M_3^{\otimes n}$ cyclically permutes these logical operators, and therefore acts as either logical $\overline{M}_3$ or $\overline{M}_3^\dagger$. 

We can achieve either possibility via the following choices of logical operators,
\begin{enumerate}
    \item For $M_3^{\otimes n} =\overline{M}_3$, choose
    \begin{equation}
        \overline{X}=\lambda_+ X^\on, \quad \overline{Y}=\lambda_+ Y^\on, \quad \overline{Z}=\lambda_+Z^\on, \quad {\rm with }~\lambda_+=(-1)^{(n-1)/2}. \label{eq:type1-phase}
    \end{equation} 
    \item For $M_3^{\otimes n} =\overline{M}_3^\dagger$, choose,
    \begin{equation}
        \overline{X}=\lambda_{-} Y^\on, \quad \overline{Y}=\lambda_{-} X^\on, \quad \overline{Z}=\lambda_-  Z^\on, \quad {\rm with }~\lambda_-=(-1)^{(n+1)/2}.  \label{eq:type5-phase}
    \end{equation}  
\end{enumerate}
\end{proof}

If $M_3^{\otimes n}$ acts as $\overline{M}_3$, we say the code is a \textbf{type 1} $M_3$-code, and if $M_3^{\otimes n}$ acts as $\overline{M}_3^\dagger$, we say the code is a \textbf{type 5} $M_3$-code. The reason for this terminology is as follows. Suppose we demand that our magic state distillation routine produce a pure $\ket{T}$ state as output, when given $n$ pure $\ket{T}$ states as input, with non-zero success probability. Then, using the fact that $M_3\ket{T}=e^{\pi i/3} \ket{T}$, we have \begin{eqnarray}
   M_3^{\otimes n} (\Pi_{\mathcal S} \ket{T}^{\otimes n}) = e^{n\pi i/3}  (\Pi_{\mathcal S} \ket{T}^{\otimes n}) = e^{n\pi i/3} \ket{\bar{T}}.  
\end{eqnarray}
We see that we must restrict attention to $[[n,1]]$ $M_3$-codes of length $n \equiv \pm 1 \pmod 6$. Furthermore, if $n \equiv 1 \pmod 6$, we should choose our logical operators such that $M_3^{\otimes n} = \overline{M}_3$, and, if $n \equiv 5 \pmod 6$, we should choose our logical operators such that $M_3^{\otimes n} = \overline{M}_3^{-1}$.

Let us now relate $W_L(\bar{r})$ to the unsigned classical weight enumerator $C(x,y)$. Suppose a logical operator $L$ of weight $n$ has sign $\lambda(L)=\pm 1$. Using Rall's rule, the sign of each $P$ in $N(\mathcal S)/\mathcal S$ as follows, 
\begin{equation}
    \lambda(P) = \begin{cases}
    \lambda(L) & \wt(P) =n \pmod 4 \\
    -\lambda(L) & \wt(P) = n+2 \pmod 4
    \end{cases}
\end{equation}
Therefore, for the choices of logical operators in Equations \eqref{eq:type1-phase} and \eqref{eq:type5-phase}, we have, 
\begin{equation}
    3 W_L(\bar{r}) =i^{-n} \lambda_\pm  C(1,i \bar{r})= 
    \begin{cases}
        -i C(1,i \bar{r})   & M_3^{\otimes n} = \overline{M}_3 \\
        +i C(1,i \bar{r}) & M_3^{\otimes n} = \overline{M}_3^\dagger. \\
    \end{cases}
\end{equation} 
Let us define $\hat{\lambda}_\pm$ which distinguishes between type 1 and type 5 codes,
\begin{equation}
    \hat{\lambda}_\pm = \begin{cases}
        -1 & M_3^{\otimes n} = \overline{M}_3 \\ 
        +1 & M_3^{\otimes n} = \overline{M}_3^\dagger.
    \end{cases}
\end{equation}

Thus, Equation \eqref{r-msd} characterizing the performance of a magic state distillation routine reduces to,
\begin{equation}
\frac{r_{\rm out}}{\sqrt{3}} = \frac{i}{3} \hat{\lambda}_\pm \frac{B(1,i\bar{r}_{\rm in})-A(1,i\bar{r}_{\rm in})}{A(1,i\bar{r}_{\rm in})}.
\end{equation}
Or, in terms of the error-rate, $\epsilon$,
\begin{equation}
\epsilon_{\rm out} = \frac{M^{(\pm)}(\epsilon_{\rm in})}{2N(\epsilon_{\rm in})} \label{error-output}
\end{equation}
where,
\begin{equation}
    M^{(\pm)}(\epsilon) =A\left(1,i \bar{r}(\epsilon)\right)-i \hat{\lambda}_\pm \frac{B(1,i\bar{r}(\epsilon))-A(1, i \bar{r}(\epsilon))}{\sqrt{3}}
\end{equation}
and,
\begin{equation}
  N(\epsilon) = A(i\bar{r}(\epsilon)).  
\end{equation}

\subsection{The 5-qubit code and other examples}
Let us now present an example that illustrates the difference between signed weight enumerators and simple weight enumerators:
\begin{example}
Consider the stabilizers $\mathcal S_{1}$ and $\mathcal S_{2}$, 
\begin{eqnarray}
\mathcal S_{1} & = & \{XX,~ZZ,~-YY,~II\} \\
\mathcal S_{2} & =& \{-XX,~-ZZ,~-YY,~II\}    
\end{eqnarray}
which are $[[2,0]]$ stabilizer codes that describe entangled stabilizer states.
The signed weight enumerators for the codes are:
\begin{equation}
W^{S_{1}}_{I}(\bar{r})=1+(2-1)\bar{r}^{2}=1+\bar{r}^{2},
\end{equation}
and,
\begin{equation}
W^{S_{2}}_{I}(\bar{r})=1-3\bar{r}^{2}.
\end{equation}
The simple weight enumerator for both the codes is: 
\begin{equation}
A(x,y)=x^{2}+3y^{3}.
\end{equation}

The stabilizer code $\mathcal S_{1}$ does not obey Rall's rule, and is therefore not an $M_{3}$-code, while $\mathcal S_{2}$ is an $M_3$-code, and therefore does obey Rall's rule.  Explicitly, $\mathcal S_{1}$ describes the state $\ket{S_1}=\frac{1}{\sqrt{2}} \left( \ket{00}+\ket{11}\right)$ and $\mathcal S_2$ describes the state $\ket{S_2}=\frac{1}{\sqrt{2}} \left(\ket{01}-\ket{10}\right)$. One can verify that $M_3\otimes M_3 \ket{S_2}=\ket{S_2}$, while $M_3 \otimes M_3 \ket{S_1} = -\frac{1}{\sqrt{2}} \left( \ket{01}+\ket{10}\right)$.
\end{example}
Next let us consider the five qubit code  \cite{MSD} and show how the noise suppression can be related to the signed and simple weight enumerators:
\begin{example}
The five qubit code is generated by the Pauli operators: \begin{equation}
+XZZXI,~+IXZZX, ~+XIXZZ, ~+ZXIXZ.
\end{equation} 
For this code, $A(1,y)=1+15y^4$, and $B(1,y)=1+30y^3+15y^4+18y^5$. As this is an $M_3$-code,
one obtains,
\begin{equation}
W_I(\bar{r})=A(1,i\bar{r}) = 1+15\bar{r}^4.
\end{equation}
Using the choice of logical operators in Equation \eqref{eq:type5-phase}, we obtain,
\begin{equation}
 \begin{split}
    3 W_L(\bar{r}) & =i\hat{\lambda}_-  C(1,i \bar{r})=i\left( 30 (i \bar{r})^3+18 (i\bar{r})^5 \right), \\
        & = 30 \bar{r}^3-18 \bar{r}^5.
    \end{split}  
\end{equation}
This gives the relation,
\begin{equation}
\begin{split}
   \epsilon_{\rm out}(\epsilon) & = \frac{15 \bar{r}^4-\sqrt{3} \left(10 \bar{r}^3-6 \bar{r}^5\right)+1}{30 \bar{r}^4+2} \\
   & = \frac{\epsilon ^2 \left(5-15 \epsilon +15 \epsilon ^2-4 \epsilon
   ^3\right)}{1-5 \epsilon +15 \epsilon ^2-20 \epsilon ^3+10 \epsilon ^4},  \label{5-qubit-distillation}
   \end{split}
\end{equation}
which agrees with Equation \eqref{msd-5}.
\end{example}

\subsection{An exhaustive search over all $[[n,1]]$ $M_3$-codes for $n<20$}

\begin{figure}
    \centering
    \includegraphics[width=0.7\linewidth]{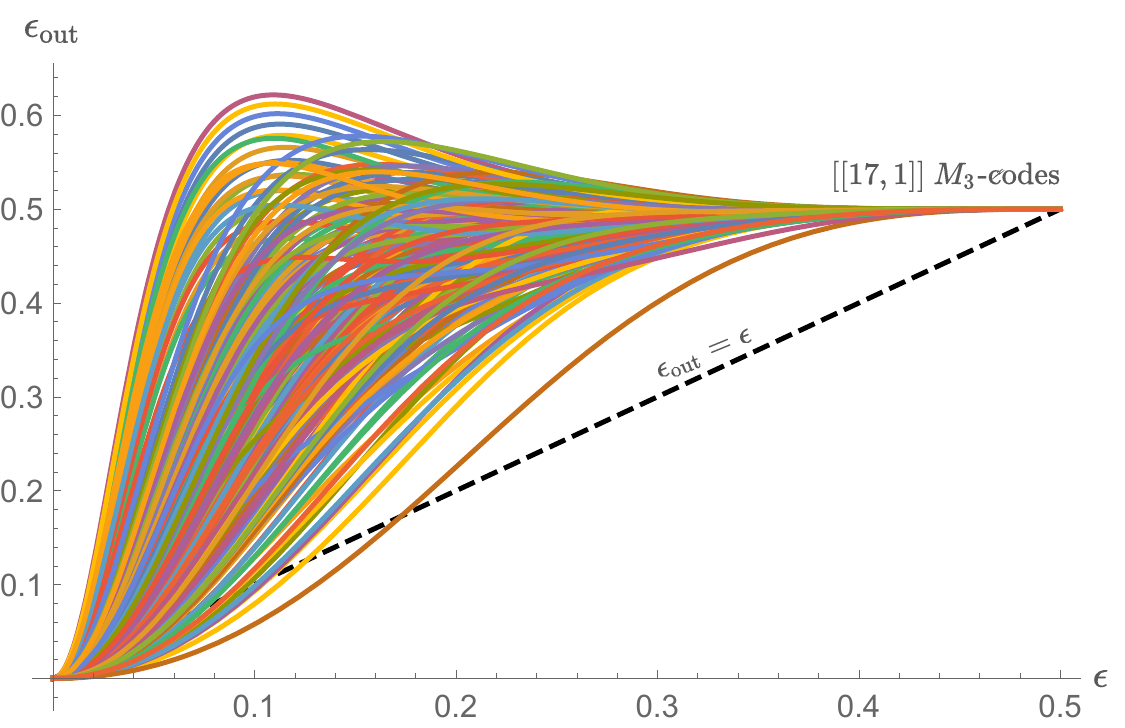}

    \includegraphics[width=0.7\linewidth]{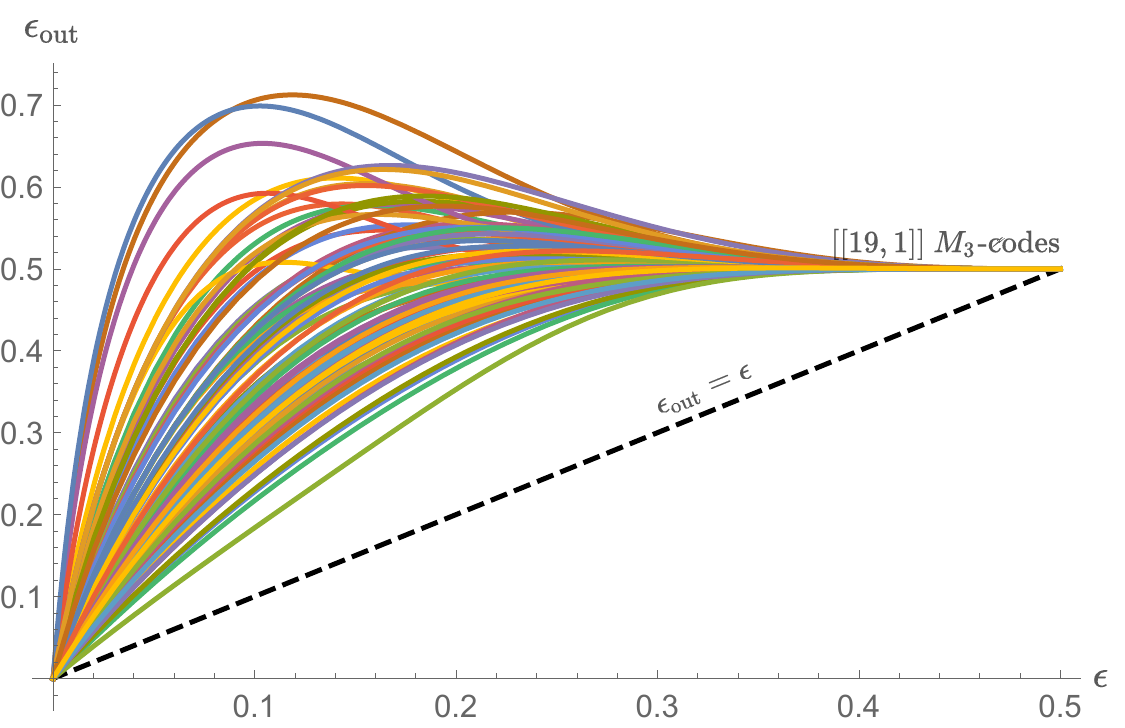}
    \caption{$\epsilon_{\rm out}=f_{\rm MSD}(\epsilon)$ computed for all $[[n,1]]$ $M_3$-codes with $n=17$ (above) and $n=19$ (below) with nonzero success probability. Solid colored lines denote $f_{\rm MSD}(\epsilon)$ computed for various $M_3$-codes, and the dashed black line is $\epsilon_{\rm out}=\epsilon_{\rm in}$. While many $[[17,1]]$ codes distill the $\ket{T}$ state with quadratic noise suppression, no code has threshold exceeding that of the 5-qubit code, $0.172673$.}
    \label{fig:search-results}
\end{figure}

Given the above results, it is straightforward to carry out a computational search for magic state distillation routines of a given length that outperform the $5$-qubit code. Here we present a search for all codes with size $n<20$. Other authors have mentioned carrying out smaller searches \cite{MSD, reichardt2005quantum} but to our knowledge, these were not claimed to be exhaustive, and \cite{rall2017signed} only searched over $M_3$-codes of length $n \leq 7$.

As mentioned above, the stabilizers of any $[[n,1]]$ $M_3$-code form a maximal self-orthogonal linear code over $GF(4)$ of odd length $n$. Any such code can be obtained by shortening a self-dual linear code over $GF(4)$ at one coordinate. All linear self-dual codes over $GF(4)$ with $n \leq 20$ have been classified in a series of papers over the past few decades \cite{macwilliams1978self, Conway1979, harada2010classification, harada2011classification}, and are conveniently available at \cite{munemasa_codes_website}. Using MAGMA \cite{magma}, we computed all inequivalent ways of shortening these codes to obtain all inequivalent $[[n,1]]$ $M_3$-codes of length $n<20$. For each such code, we computed the weight enumerator and then the threshold using Equation \eqref{error-output}. { The computation took less than a few minutes on an ordinary desktop computer.}

Results for $n=17$ and $n=19$ using codes classified only somewhat recently (in 2010 and 2011) by Harada et. al.  \cite{harada2010classification, harada2011classification} are shown in Figure \ref{fig:search-results}. We found that no codes have a threshold which exceeds that of the 5-qubit code.

While we hope to extend this computational search to codes of larger length in the near future, in the next section, we will attempt to use Equation \eqref{error-output} to place more general constraints on magic state distillation routines.

\section{Quantum constraints on weight enumerators}
\label{sec:constraints}

Any linear self-orthogonal code over $GF(4)$ defines an $M_3$ stabilizer code. In particular, a self-dual code over $GF(4)$ defines an $[[n,0]]$ $M_3$-state; and a maximal self orthogonal code, defines an $[[n,1]]$ $M_3$-code. It is a celebrated result of classical coding theory that the weight enumerators of such codes are not arbitrary, and must satisfy certain non-trivial constraints, which arise from (1) demanding invariance under MacWilliams transform, and (2) demanding non-negativity of the coefficients of the weight enumerator. We refer to such constraints as ``classical constraints'', and we review them in detail in the next subsection.  

We will then derive a new family of ``quantum constraints'' on these weight enumerators, which arise from demanding consistency with their physical interpretation as potential magic state distillation routines. These constraints are independent of classical constraints, and severely constrain the space of consistent weight enumerators. 

\subsection{Review of classical constraints on weight enumerators}
\label{sec:invariant-theory}

\subsubsection{Self-dual codes}
The weight enumerator of a linear self-dual code over $GF(4)$, $A(x,y)$ must be invariant under the MacWilliams transform, $A(x,y) = A(\frac{x+3y}{2}, \frac{x-y}{2})$. Additionally, $n$ must be even, so $A(x,y)=A(-x,-y)$, and all codewords must have even weight, so $A(x,y)=A(x,-y)$.
The theory of polynomial invariants \cite{sloane1977error, macwilliams1978self, nebe2006self} dictates that the most general polynomial $A(x,y)$ invariant under these operations is given by the following Theorem.
\begin{theorem} \label{thm:self-dual-invariant-theory}
    The weight enumerator of an $[[n,0]]$ $M_3$-state must be an element of the ring $\mathbb{C}[f,g]$, generated by the primary invariants:
\begin{align}
    f &= x^2 + 3y^2, \\
    g &= y^2(x^2-y^2)^2.
\end{align}
Explicitly, 
\begin{equation}
    A(x,y) = \sum_{j=0}^{\floor{\frac{n}{6}}} c_j {f}(x,y)^{\frac{n}{2}-3j} {g}(x,y)^j, \label{eq:invariant-theory-self-dual}
\end{equation}
\end{theorem}
This reduces the dimensionality of the space of weight enumerators for a code of length $n$ from $n/2+1$ to $ \lfloor n/6 \rfloor + 1$. For example, for $n=12$, the space of all possible weight enumerators is a 3-dimensional space spanned by $c_0$, $c_1$ and $c_2$. Of course, we can always set $c_0=1$ by demanding $A_0=1$.

However, invariant theory only guarantees that the polynomial has the correct symmetries. To represent a valid code, the coefficients $A_w$ (which count the number of codewords of weight $w$) must be non-negative integers. One can extract $A_w$ by expanding Equation \eqref{invariant-theory-self-dual} in powers of $y$:
\begin{equation}
    A_{w} = \frac{1}{w!} \left[\frac{d^w}{dy^w}\sum_{j=0}^{\floor{n/6}} c_j f^{n/2-3j}g^j\right]_{y=0} \geq 0. \label{eq:classical-constraints}
\end{equation}
We refer to these constraints as \textbf{classical constraints}, and the region of parameter space $\{c_j\}$ satisfying these constraints and the \textbf{classical feasible region}. 

Similar expressions can be derived for various families of self-dual and self-orthogonal codes; one then attempts to find candidate weight enumerators for codes with maximum distance that satisfy these classical constraints using linear programming techniques \cite{MALLOWS1973188, macwilliams1978self, conway1990new, rains,  rains2003new}.

\subsubsection{Maximal self-orthogonal codes}
It is also possible to constrain the weight enumerators of maximal self-orthogonal codes using invariant theory, as demonstrated by Mallows and Sloane \cite{mallows1974weight, mallows1981weight}. For $[[n,1]]$ $M_3$-codes, which correspond to maximal self-orthogonal $GF(4)$ codes, the expression for $A(x,y)$ in terms of invariants is derived in Chapter 10 of the monograph by Nebe, Rains and Sloane \cite{nebe2006self}, and is slightly more complicated.
\begin{theorem}(Theorem 10.5.1 of Nebe, Rains and Sloane \cite{nebe2006self})
The simple weight enumerator of any $[[n,1]]$ $M_3$-code must be of the form:
\begin{equation}
A_{MSO}(x,y) = x \underbrace{\sum_{j=0}^{\floor{(n-1)/6}}c_j'  f(x,y)^{\frac{n-1}{2}-3 j}  g(x,y)^{j}}_{S_{1}} + xy^2 (x^2-y^2) \underbrace{\sum_{j=0}^{\floor{(n-5)/6}} d_j'  f(x,y)^{\frac{n-5}{2} -3j}  g(x,y)^{j}}_{S_{2}}, 
\label{Aeq}
\end{equation}
where the $c_j'$ and $d_j'$ are arbitrary coefficients. \label{invariant-theorem}
\end{theorem}
Note that, via the MacWilliams identity, this theorem also determines the weight enumerator of $N(\mathcal S)$,  $B(x,y)$ 
to be, 
\begin{equation}
    B(x,y) = \frac{1}{2^{n-1}} A(x+3y,x-y) = (x+3y)  S_1 +2^{-1} y (x-y) (x+3 y)(x^2-y^2) S_2.
\label{Beq}
\end{equation}
The weight-enumerator for logical operators $C(x,y)=B(x,y)-A(x,y)$, is therefore,
\begin{dmath}
    C(x,y)=B(x,y)-A(x,y) = 3y S_1 +2^{-1} y (x^2 - y^2)(x^2 - 3 y^2) S_2.
\label{Ceq}
\end{dmath}

For maximal self-orthogonal codes, invariant theory provides weaker constraints than for self-dual codes -- the space of consistent weight enumerators has dimensionality $\floor{(n-1)/6}+\floor{(n-5)/6}+2 \sim n/3$. However, there are more classical non-negativity constraints, as the coefficients $A_w$ and $C_w$ of both $A(x,y)$ and $C(x,y)$ must be non-negative. 

Before we continue, let us pause to remark that, using elementary considerations, we can fix some of the unknown coefficients $c_j'$ and $d_j'$ in $S_1$ and $S_2$. Note that, for any code, $A_0=1$ ({ because the identity operator is always in the stabilizer of the code}), which implies $c_0'=1$. We demand $C_1=0$, because any $[[n,1]]$ $M_3$-code with a weight $1$ logical operator must be the tensor product of the identity operator with a $[[n-1,0]]$ stabilizer state, and is therefore trivial. This implies $d_0'=-6$. We also demand that $A_2 \neq 0$ because any $M_3$ code with a weight-2 stabilizer must be equivalent to a code that contains the stabilizers $-X\otimes X \otimes I^{\otimes n-2}$, { $-Y\otimes Y \otimes I^{\otimes n-2}$}, and $-Z\otimes Z \otimes I^{\otimes n-2}$; and is therefore a tensor product of the two-qubit $M_3$-state with another $[[n-2,1]]$ $M_3$-code, so is trivial. Demanding $A_2=0$ gives $c_1'=\frac{3}{2}(5-n)$.

\subsection{Quantum constraints}
\label{sec:quantum-consistency}

While any point in the classical feasible region represents a mathematically plausible list of codeword weights, from the perspective of classical combinatorics, not all such lists correspond to physically realizable quantum stabilizer codes.
Because any code can be used to define a magic state distillation protocol, it must satisfy additional \textbf{quantum consistency conditions}.

\begin{enumerate}
    \item \textbf{Non-negative success probability:} The probability of successfully projecting $n$ copies of a twirled  noisy magic state $\rho_T(\epsilon)^{\otimes n}$ onto the code space of any $M_3$-code must be non-negative for all physical error rates $\epsilon$. This success probability is given by,
    \begin{equation}
        \eta(\epsilon) = W_I(\bar{r}(\epsilon)) = A(1, i\bar{r}(\epsilon)) \geq 0, \quad \forall \epsilon \in [0, 1].
    \end{equation}
    This is a powerful constraint because it applies when the formal variable $y=i\bar{r}$ in the weight enumerator is imaginary. A polynomial $A(1, y)$ with positive coefficients can become negative when evaluated at imaginary $y$.

    \item \textbf{Threshold bound:} As discussed in Section \ref{sec:basics}, for distillation via an $[[n,1]]$ $M_3$-code, if the input noisy magic states lie within the stabilizer octahedron ($\epsilon_{\rm in} > \epsilon_{\rm max} \approx 0.21$), then $\epsilon_{\rm out}>\epsilon_{\rm max}$. \footnote{This constraint follows directly from the Gottesman-Knill theorem -- mixtures of stabilizer states must be mapped to mixtures of stabilizer states under Clifford unitaries and stabilizer measurements. Its validity does not rely on the natural, but unproven, assumption that non-stabilizer states cannot be simulated efficiently. } A valid quantum code must exhibit thresholds $\epsilon_* \leq \epsilon_{\rm max}$.
\end{enumerate}

We find that these quantum constraints are independent of classical constraints and drastically reduce the region of feasible codes. We will illustrate this by explicitly working out all classical and quantum constraints for two instructive examples:
\begin{enumerate}
    \item the space of $[[12,0]]$ $M_3$-states (which correspond to self-dual $[12,6]_{\rm GF(4)}$ codes); and, 
    \item the space of $[[11,1]]$ $M_3$-codes (which correspond to maximal self-orthogonal $[11,5]_{\rm GF(4)}$ codes.)
\end{enumerate}
In Appendix \ref{app:small-examples}, we perform similar computations to constrain $[[n,0]]$ and $[[n,1]]$ $M_3$-codes for other small values of $n$.

\subsection{Constraining $[[12,0]]$ $M_3$-states} 
\label{sec:12-0}
For the case of $n=12$, we find from Equation \eqref{eq:invariant-theory-self-dual}, that invariant theory implies that a general weight enumerator depends on three coefficients, $c_0$, $c_1$ and $c_2$. Using $A_0=1$ to fix $c_0$, we find, 
\begin{equation}
\begin{aligned}
    A(x, y) = & \, x^{12} + x^{10}y^2 (18 + c_1) + x^{8}y^4 (135 + 7c_1 + c_2) + x^{6}y^6 (540 + 10c_1 - 4c_2) \\
              & + x^4y^8 (1215 - 18c_1 + 6c_2) + x^2y^{10} (1458 - 27c_1 - 4c_2) + y^{12} (729 + 27c_1 + c_2).
\end{aligned}
\end{equation}
From this expression, we read off the classical constraints, $A_{w} \geq 0$, for $w=2, \ldots, 12$:
\begin{eqnarray}
    18 + c_1 &\geq & 0, \\
    135 + 7c_1 + c_2 &\geq & 0, \\
    540 + 10c_1 - 4c_2 &\geq & 0, \\
    1215 - 18c_1 + 6c_2 &\geq & 0, \\
    1458 - 27c_1 - 4c_2 &\geq & 0, \\
    729 + 27c_1 + c_2 &\geq & 0. 
\end{eqnarray}
The convex polytope allowed by these classical constraints in $c_1-c_2$ plane is shown in blue in Figure \ref{fig:12-0-codes}.

\begin{figure}
    \centering
    \includegraphics[width=\linewidth]{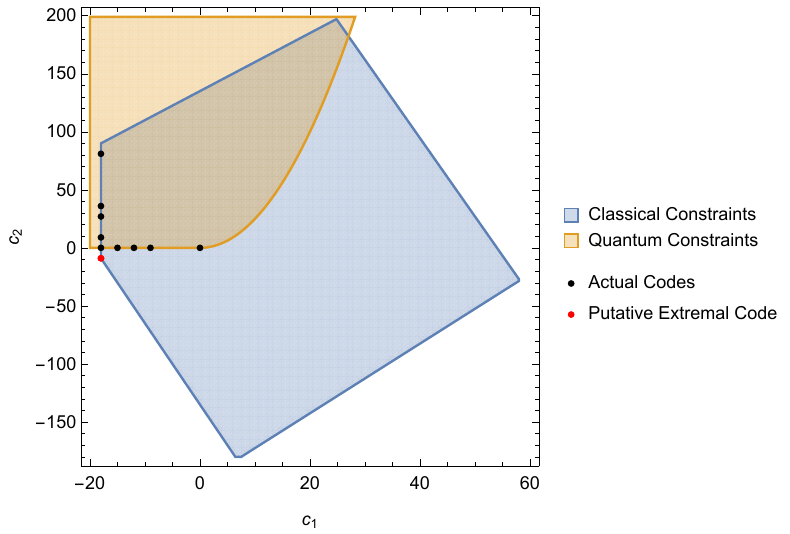}
    \caption{Linear programming bounds for classical self-dual linear  $[12,6]_{GF(4)}$ codes. The set of allowed weight enumerators is parameterized by two real numbers, $c_1$ and $c_2$. The region of the $c_1$-$c_2$ plane allowed by classical  linear programming constraints is shown in blue. The region allowed by the quantum  constraints $A(1,i\frac{1-2\epsilon}{\sqrt{3}})\geq 0$ for all $0 \leq \epsilon\leq 1$ is shown in orange.  Black points denote all self-dual $[12,6]_{GF(4)}$ linear codes. The red point denotes the putative extremal $[12,6,6]_{GF(4)}$ weight enumerator which does not correspond to an actual self-dual linear code. This code, though allowed by classical constraints, is ruled out by quantum constraints.}
    \label{fig:12-0-codes}
\end{figure}

We now consider the quantum constraint that the success probability be non-negative, which implies 
\begin{equation}
    A(1,\frac{i}{\sqrt{3}}(1-2\epsilon))\geq 0, \label{success-prob}
\end{equation} 
for all $\epsilon \in [0,1]$. This infinite collection of linear constraints can be translated into relatively simple conditions on $c_j$, as follows. $f(1,\frac{i}{\sqrt{3}}(1-2\epsilon)) = 4\epsilon(\epsilon-1) \to 0$ as $\epsilon \to 0$. Therefore, demanding the success probability be non-negative for pure magic states, we have a constraint only on $c_{\floor{n/6}}$:  
\begin{equation}
(-1)^{\floor{n/6}} c_{\floor{n/6}} \geq 0.
\end{equation} 
For non-zero $\epsilon$, define $\phi(\epsilon)=-g/f^3$; and observe that $\phi \in [0, \infty]$ for $\epsilon \in [0,1]$. The constraint \eqref{success-prob} then becomes 
\begin{equation}
    1+\sum_{j=1}^{\floor{n/6}}\phi^j (-1)^jc_j \geq 0, \quad \forall \phi \in [0,\infty).
\end{equation}
For $n=12$, we have $1-c_1 \phi + c_2 \phi^2 \geq 0$, which reduces to
\begin{equation}
    c_2 \geq \begin{cases}
         0 & c_1 <0 \\
         c_1^2/4 & c_1 \geq 0.
    \end{cases}
\end{equation}

We plot the quantum constraint in orange in Figure \ref{fig:12-0-codes}. A computer search \cite{macwilliams1978self} classified all actual $[12,6]_{GF(4)}$ codes.  These codes are available on \cite{munemasa_codes_website}, and have 9 distinct weight enumerators, which are also shown in Figure \ref{fig:12-0-codes}. Figure \ref{fig:12-0-codes} shows that the quantum constraints are independent of classical constraints, and substantially reduce the region of feasible weight-enumerators. The number of consistent integral weight enumerators that satisfy the classical constraints is 2919; of these, only 570 also satisfy the quantum constraints.

Interestingly, this result resolves a long-standing (albeit slightly obscure) mystery in classical coding theory.  The weight enumerator with maximum possible distance consistent with Theorem \ref{thm:self-dual-invariant-theory} is unique, and can be computed to be \cite{macwilliams1978self},
\begin{equation}
   A(1,y) = 1+396 y^6+1485 y^8+1980 y^{10}+234 y^{12}. \label{extremal-12}
\end{equation}
This weight enumerator for a putative $[12,6,6]_{GF(4)}$ code has positive integral coefficients divisible by 3. Moreover, as shown in Figure \ref{fig:12-0-codes}, it lies at an extremal vertex of the classical linear programming constraints $A_i\geq 0$. Several authors therefore conjectured that a self-dual code $[12,6,6]_{GF(4)}$ must exist \cite{semakov1971uniformly,paper2} -- although no explicit construction of such a code was known. Computer searches reported by MacWilliams, Odlyzko, Sloane and Ward \cite{macwilliams1978self} later revealed that, in fact, no $[12,6,6]_{GF(4)}$ code exists. For many decades, the non-existence of this code was simply an inexplicable combinatorial fact, demonstrable only via a tedious brute-force search. We now see from Figure \ref{fig:12-0-codes}, that this code is ruled out by the quantum constraints. In particular, the projection probability for pure magic states arising from Equation \eqref{extremal-12} is negative:
\begin{equation}
A(1, \frac{i}{\sqrt{3}}) = 1-\frac{396}{3^3}+\frac{1485}{3^4}-\frac{1980}{3^5}+\frac{234}{3^6} = -(256/81) <0,
\end{equation}
providing a one-line proof of the nonexistence of the extremal $[12,6,6]_{GF(4)}$ code. In section \ref{sec:self-dual}, we generalize this result to rule out all extremal self-dual linear $GF(4)$ codes with lengths divisible by $12$.

\subsection{Constraining $[[11,1]]$ $M_3$-codes} 
\label{sec:11-1}
From Equation \eqref{Aeq}, weight enumerators for $[[11,1]]$ $M_3$-codes depend on four coefficients $c_0'$,$c_1'$, $d_0'$ and $d_1'$. We reduce this to two coefficients, by demanding that $A_0=1$ and $C_{1}=0$. (Any $M_3$-code with $C_1 \neq 0$ is necessarily trivial.) To impose the classical linear programming constraints, it suffices to focus on $B(x,y)$, as it contains contributions from stabilizers (which are of even weight) and logical operators (which are of odd weight). We find from equation \eqref{Beq},
 \begin{align*}
   B(x,y) = &x^{11} + x^9y^2 \left( c'_1 + 9 \right) 
    + \frac{1}{2} x^8y^3 \left( 6c'_1 + d'_1 + 60 \right) 
    + x^7y^4 \left( 4c'_1 + d'_1 + 42 \right) \\
    &+ 3x^6y^5 \left( 4c'_1 - d'_1 + 96 \right) 
    + x^5y^6 \left( -2c'_1 - 3d'_1 + 162 \right) 
    + x^4y^7 \left( -6c'_1 + 6d'_1 + 972 \right) \\
    &+ 3x^3y^8 \left( -4c'_1 + d'_1 + 135 \right) 
    + x^{2}y^9 \left( -36c'_1 - 5d'_1 + 1296 \right) 
    + y^{10} \left( 9c'_1 - d'_1 + 405 \right) \\
    &+ \frac{3}{2} y^{11} \left( 18c'_1 + d'_1 + 324 \right).
\end{align*}
As before, we plot the classical constraints arising from demanding $B_w \geq 0$ in blue in Figure \ref{fig:11-linear-programming}, along with all $[[11,1]]$ $M_3$-codes (obtained from shortening the self-dual codes in Munemasa's database \cite{munemasa_codes_website}).

Let us now turn to the quantum constraints. We impose the constraint of non-negative success probability as before; the region allowed by this constraint is shown in green in Figure \ref{fig:11-linear-programming}. To impose the threshold constraint, we compute $f_{\rm MSD}(\epsilon)$ using Equation \eqref{error-output}, and then demand that $f_{MSD}(\epsilon_{\rm in}) \geq  \epsilon_{\rm max}$ for all $\epsilon_{\rm in} \in [\epsilon_{\rm max},1-\epsilon_{\rm max}]$ for both choices of logical operators $M_3^{\otimes n}=\overline{M}_3$ or $M_3^{\otimes n}=\overline{M}_3^\dagger$. We carry this out numerically, for a large number of randomly chosen values of $\epsilon_{\rm in} \in [\epsilon_{\rm max},1-\epsilon_{\rm max}]$. The region satisfying the distillation threshold constraint is shown in orange in Figure \ref{fig:11-linear-programming}.

From Figure \ref{fig:11-linear-programming} we see that the two quantum constraints are independent of each other and the classical constraints. They serve to tightly constrain the space of allowed weight enumerators to a remarkably small region.
  
\begin{figure}
    \centering
    \includegraphics[width=0.95\linewidth]{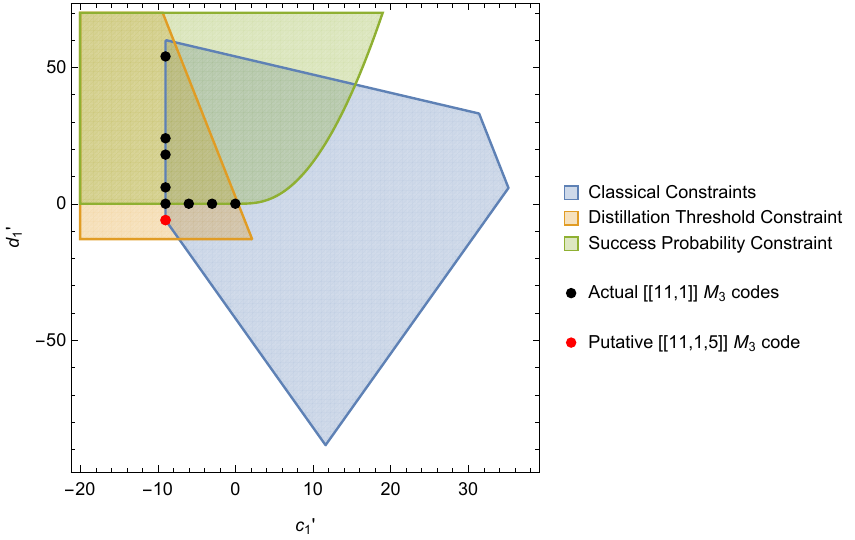}
    \caption{Linear programming constraints for $n=11$ magic state distillation routines, restricted to the plane defined by $C_1=0$. Regions are labeled as in Figure \ref{fig:7-1-linear-programming}. Black points denote all $[[11,1]]$ $M_3$-codes. The red point denotes the only putative $[[11,1,5]]$ $M_3$-code allowed by classical constraints -- we see it is ruled out by the new quantum constraints.}
    \label{fig:11-linear-programming}
\end{figure}

As in the case of $12$-qubit $M_3$-codes, the quantum constraints also serve to rule out a putative high-distance code whose weight enumerator is consistent with all classical constraints. We can define a unique weight enumerator corresponding to an $[[11,1,5]]$ $M_3$-code, which lies at the corner of the blue region in Figure \ref{fig:11-linear-programming}, corresponding to $(c_1',d_1')=(-9, -6)$. This has a weight enumerator with all positive coefficients, 
\begin{equation}
    B(1,y) = 1+198 y^5+198 y^6+990 y^7+495 y^8+1650 y^9+330
   y^{10}+234 y^{11},
\end{equation}
so one might have expected that it corresponds to a valid $[[11,1,5]]$ $M_3$-code. However, it is ruled out by the quantum constraints, which imply that the maximum distance for an $[[11,1]]$ $M_3$-code is $4$. We generalize this observation to place stronger upper bounds on the distance of $M_3$-codes of arbitrary length in Section \ref{sec:odd-n-bounds}.

While our new quantum constraints are quite powerful, they still leave a slight mystery.  The region allowed by classical and quantum constraints is \textit{almost} as small as it can be, i.e., it is almost equal to the convex hull of the set of all weight enumerators of actual $[[11,1]]$ $M_3$-codes. The difficulty is that the top left corner of the blue region in Figure \ref{fig:11-linear-programming}, $(c_1',d_1')=(-9, 60)$, lies within the quantum constraints and corresponds to a consistent weight enumerator for an $M_3$-code, 
\begin{equation}
    B(1,y)=1+33 y^3+66 y^4+1386 y^7+693 y^8+1320 y^9+264
   y^{10}+333 y^{11}.
\end{equation} However, from the classification in Munemasa's database \cite{munemasa_codes_website}, no $[[11,1]]$ $M_3$-code exists that has this weight enumerator. Computing the threshold of this weight enumerator using Equation \eqref{error-output}, we find it would have had a threshold of $0.190827$, slightly exceeding that of the 5-qubit code. It would be interesting to understand ``why'' this code does not exist.

\section{Quantum bounds for classical self-dual codes}
\label{sec:self-dual}

As discussed in the previous section, Theorem \ref{thm:self-dual-invariant-theory} states that the weight enumerator of any linear self-dual  $GF(4)$ code must take the form \cite{macwilliams1978self, nebe2006self},
\begin{equation}
    A(x,y) = \sum_{j=0}^{\floor{\frac{n}{6}}} c_j {f}(x,y)^{\frac{n}{2}-3j} g(x,y)^j, \label{invariant-theory-self-dual}
\end{equation}
and is known as a Gleason's theorem \cite{gleason1971weight, macwilliams1972generalizations}. Classical linear self-dual $GF(4)$ codes are one of the four celebrated types of classical self-dual codes for which a Gleason's theorem applies \cite{rains2002self}. This theorem immediately places an upper bound on the distance of a self-dual code, commonly known as the Mallows-Sloane bound \cite{MALLOWS1973188}. The generalized Mallows-Sloane bound on the distance $d$ of a Hermitian self-dual code over $GF(4)$ of length $n$ was established by MacWilliams et al.~\cite{macwilliams1978self} as we now review.

After demanding $A_0=1$, the weight enumerator for a self-dual code depends on $\floor{n/6}$ parameters: $c_1$, $c_2$, $\ldots c_{m}$. If we demand $A_{2j}=0$ for $j=1,\ldots r$, we have exactly $\floor{n/6}$ linear equations in $\floor{n/6}$ unknowns, which determine the $\{c_i\}$ uniquely. The explicit solution is obtained by MacWilliams et. al.  \cite{macwilliams1978self} to be,
\begin{equation}
    c_j = \frac{n}{2j} \sum_{r=0}^{j-1} (-3)^{r+1} \binom{n/2-3j+r}{r} \binom{3j-r-2}{j-r-1}. \label{extremal}
\end{equation}
The weight enumerator $A(x,y)$ defined by Equation \eqref{extremal} is said to be an \textbf{extremal} weight enumerator, and is unique for any $n$. One can check that $A_{2\floor{n/6}+2} \neq 0$ for these extremal enumerators, so the distance of any self-dual code is bounded from above by 
\begin{equation}
    d \leq 2 \left\lfloor \frac{n}{6} \right\rfloor + 2.
    \label{eq:mallows-sloane}
\end{equation}
For lengths $n=12m$, this bound permits a distance of $d = 4m+2$.

Using equation \eqref{extremal}, consistent extremal weight enumerators with positive integral $A_w$ divisible by 3 exist for all even $n \leq 96$. A central question in coding theory is determining whether or not self-dual codes over $GF(4)$ realizing these extremal weight enumerators exist. Such codes, if they exist are known as extremal self-dual codes (of type $4^H$). 
MacWilliams \textit{et al.} classified type $4^H$ extremal codes of small length; extremal codes were found for all even $n \leq 18$ except $n=12$. In subsequent decades, researchers, notably Huffman, Lam and Pless (\cite{lam1990there,huffman1990extremal, huffman1991extremal}) attempted to classify larger lengths, and established the nonexistence of an extremal $[24,12,10]_{ GF(4)}$ code. The existence of extremal self-dual codes of lengths $12m$ for $m>2$ is an open question. 

We saw in the previous section that the new quantum constraint of non-negative success probability, provides a one-line proof of the non-existence of the $[12,6,6]$ extremal code. We extend this observation into the following theorem, which applies to all codes with length divisible by $12$.  
\begin{theorem}
No extremal type $4^H$ self-dual code of length $12m$ 
exists. \label{thm: no-self-dual}
\end{theorem}
\begin{proof}
To obtain the extremal weight enumerator  of length $n=12m$, one determines the $2m-1$ unknowns in Equation \eqref{invariant-theory-self-dual}, $\{c_{2},$ $c_{3},$ $\ldots, c_{2m}\},$ from the $2m-1$ equations, $A_2=A_4=\ldots = A_{4m}=0$. The resulting weight enumerator  corresponds to a putative code with parameters $[12m,6m, 4m+2]_{GF(4)}$, with $c_j$ given by Equation  \eqref{extremal}. 
Specializing to $n=12m$, one can easily compute that, 
$$c_{2m} = -\sum_{t=0}^{m-1}3^{2+2t} \frac{(4 t+1) (6 m-2 t-3)!}{(4 m-1)! (2 m-2 t-1)!},$$
which is manifestly negative.
For a code of length $12m$, using Equation \eqref{invariant-theory-self-dual}, we find $A(1, \frac{i}{\sqrt{3}})=(-16/27)^{2m}c_{2m}$. Because $c_{2m}$ is strictly negative and $(-16/27)^{2m}$ is strictly positive,  $A(1, \frac{i}{\sqrt{3}})$ must be negative. However, as established in Section \ref{sec:quantum-consistency}, $A(1, \frac{i}{\sqrt{3}})$ represents the probability of successfully projecting $n$ pure magic states onto the stabilizer state corresponding to the putative code; and therefore must be non-negative. Thus, by contradition, no extremal self-dual code with length divisible by 12 can exist.\end{proof} 

This theorem not only explains the non-existence of  extremal type $4^H$ self-dual codes of length $12$ and $24$, it also rules out extremal codes of lengths $36$, $48$, $60$, $72$, $84$ and $96$ -- the existence of which was, to the best of our knowledge, previously an open question (see \cite{macwilliams1978self}). For other values of $n$ we find $A(1, \frac{i}{\sqrt{3}})$ is non-negative, and so we cannot rule out their existence. 

\subsection{Bounds on distance for $[[n,1]]$ $M_3$-codes}
\label{sec:odd-n-bounds}
As a corollary of Theorem \ref{thm: no-self-dual}, our quantum constraints place a stronger upper bound on the distance of $[[n,1,d]]$ $M_3$-codes for $n \equiv 11 \pmod {12}$ than those previously known in the literature \cite{rains}. In particular, Rains \cite{rains} obtained upper bounds for the distance $d$ of a general $[[n,1,d]]$ stabilizer code using linear programming. These bounds continue to hold without modification for $M_3$-codes, and are summarized in the following theorem.
\begin{theorem}[Rains, \cite{rains}]
    Demanding only non-negativity of $A_i$, $B_i$ and $C_i$ implies that the distance of an $[[n,1,d]]$ $M_{3}$-code satisfies $d<d_{\max}$, with  \begin{equation}
    d_{\text{max}} \leq \begin{cases}
        2m+1  & n = 6m+1, ~
        n=6m+3 \\
       2m  +3  & n = 6m+5.
    \end{cases}
\end{equation}
\label{thm:old-theorem}
\end{theorem}

Using Theorem \ref{thm: no-self-dual}, we can strengthen Theorem \ref{thm:old-theorem} to the following result.
\begin{theorem}
\label{thm:distance-bound}
The quantum distance of a non-degenerate $[[n,1,d]]$ $M_3$-code satisfies:
\begin{equation}
    d < 2 \lfloor \frac{n+1}{6} \rfloor + 1.
\end{equation} \label{thm:distance-M_3-code}
\end{theorem}
\begin{proof}
To translate Theorem \ref{thm: no-self-dual} to a bound on the quantum distance of an $[[n,1]]$ $M_3$-code, we use the fact that any $[[n,1]]$ stabilizer code $S$  can be uniquely extended to $[[n+1,0]]$ stabilizer state $\mathcal S_{\rm ext}$. Explicitly, the stabilizers of $\mathcal S_{\rm ext}$ are constructed from the stabilizers and logical operators of $\mathcal S$ as follows: 
\begin{itemize}
    \item For each stabilizer $P$ of $\mathcal S$, construct a stabilizer $P_{\rm ext}$ of $\mathcal S_{\rm ext}$ of the form $P \otimes I$. Note that $\wt~ P_{\rm ext}=\wt ~ P$.
    \item For each representative of the logical operator $L$, $P_L \in \mathcal N(S)$, construct a stabilizer $P_{\rm ext}'$ of $\mathcal S_{\rm ext}$ of the form $P_L \otimes L$. Note that $\wt~P_{\rm ext}'=\wt~ P_L+1$.
\end{itemize}
Since logical operators in an $[[n,1]]$ $M_3$-code have odd weight, $d+1$ is even. Thus, if the quantum distance of a \textit{non-degenerate} $[[n,1]]$ $M_3$-code is $d$, the distance of the extended $[[n+1,0,d_{\rm ext}]]$ state is at least $d+1$. Note that $d_{\rm ext}$ is the distance of the corresponding classical self-dual linear $[n+1,\frac{n+1}{2}]_{GF(4)}$ code,

The existence of a non-degenerate $[[12m-1,1]]$ $M_3$-code with distance $d= 4m+1$ would imply a self-dual code of length $12m$ with distance $4m+2$.
Theorem \ref{thm: no-self-dual} rules out such self-dual codes.
Therefore, no non-degenerate $[[12m-1, 1,4m+1]]$ $M_3$-code with quantum distance $4m+1$ can exist.
\end{proof}
Our proof of this bound applies only to non-degenerate $M_3$-codes. A degenerate $[[n,1,d]]$ code of distance $d$ contains a stabilizer of weight $w<d$; the distance of the extended code would be also be $w<d+1$, and therefore need not violate Theorem \ref{thm: no-self-dual}.  However, when we attempted to solve the linear programming problem in the same manner as we did for $n=11$ in section \ref{sec:11-1}, using a computer algebra system for finite values of $n$ up to $215$, allowing for the possibility of degeneracy, we found that Theorem \ref{thm:distance-M_3-code} remains true. We therefore conjecture that Theorem \ref{thm:distance-bound} applies to all degenerate and non-degenerate $M_3$-codes.

We compare Theorem \ref{thm:distance-bound} to the earlier bound of Rains \cite{rains} in Figure \ref{fig:distance-bounds}.  We emphasize that theorem \ref{thm:distance-M_3-code} is a consequence of the new quantum constraints -- we checked that  $[[12m-1,1,4m+2]]$ $M_3$-codes are not ruled out by purely classical constraints. We find it is possible to obtain even stronger distance bounds for larger values of $n$, by solving the linear programming problem in a computer algebra system. These are presented in Appendix \ref{app:distance}.

We remark that we can  apply a similar argument to bound the \textit{classical distance} $d_C$ of maximal self-orthogonal codes over $GF(4)$ of length $12m-1$, i.e., the minimum weight of a stabilizers of the $[[n,1]]$ $M_3$-code. 

\begin{theorem}
\label{thm:classical-bound}
Let $\mathcal{C}$ be a maximal self-orthogonal linear code over $GF(4)$ of length $n = 12m-1<215$. The minimum Hamming distance $d$ of $\mathcal{C}$ satisfies
\begin{equation}
    d \leq 4m,
\end{equation}
which is strictly less than the classical bound of $4m+2$.
\end{theorem}
Again this was proven using linear programming, as done explicitly in section \ref{sec:11-1} for $n=11$, using a computer algebra system for $n\leq 215$. We conjecture that it holds for arbitrary $n$.

\begin{figure}
    \centering
    \includegraphics[width=0.95\linewidth]{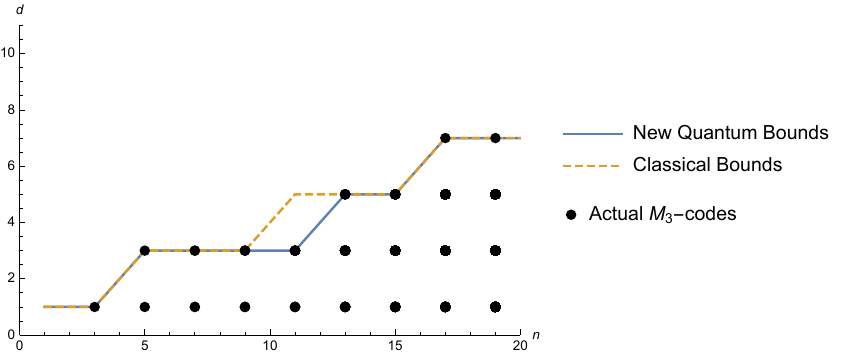}
\caption{We plot our quantum bounds on the distance of an $[[n,1,d]]$ $M_3$-code from Theorem \ref{thm:distance-M_3-code} (in solid blue) and the classical bounds from Rains \cite{rains} (in dashed orange) for codes of size less than $20$. We also plot distances of all  $[[n,1,d]]$ $M_3$-codes, which have been classified for $n \leq 19$ as black points, showing that the quantum bounds are saturated.
    }
    \label{fig:distance-bounds}
\end{figure}

\section{Constraining distillation routines}\label{sec:distillation-routines}

In this section, we use invariant theory and linear programming to constrain the performance of a magic state distillation routine based on an $M_3$-code. As a consequence of section \ref{sec:quantum-consistency}, we cannot use linear programming to directly constrain the threshold. We can, however, constrain the noise suppression exponent $\nu$ of a distillation routine, which characterizes the noise suppression of a distillation routine for small $\epsilon$ via,
\begin{equation}
      f_{\rm MSD}(\epsilon) = \Theta(\epsilon^\nu),  
\end{equation} 
and plays a role for distillation analogous to the role played by the distance $d$ for conventional quantum error-correction. Here we show that it is possible to constrain the noise-suppression exponent using invariant theory and linear programming, using techniques similar to those used to constrain the code distances \cite{MALLOWS197568,rains,MALLOWS1973188,conway1990new}. 

\subsection{Results from invariant theory}

We now translate the invariant expressions for weight enumerators of an $[[n,1]]$ $M_3$-code to expressions for the distillation performance of the corresponding distillation routine. Recall that the weight enumerator for stabilizers $\mathcal S$, is given by $A(x,y)=xS_1+xy^2(x^2-y^2)S_2$, and the weight enumerator for logical operators $N(\mathcal S)/S$ is given by $C(x,y)= 3y S_1 +2^{-1} y (x^2 - y^2)(x^2 - 3 y^2) S_2$, where,\begin{eqnarray}
S_1 & = & \sum_{j=0}^{\floor{(n-1)/6}}c_j'  f^{\frac{n-1}{2}-3 j}  g^{j},\\ S_2  &=&  \sum_{j=0}^{\floor{(n-5)/6}} d_j'  f^{\frac{n-5}{2} -3j}  g^{j}.
\end{eqnarray}
and $f$ and $g$, expressed in terms of $\epsilon$, are
\begin{equation}
f(1,i\bar{r}(\epsilon)) = 4 (1-\epsilon) \epsilon, \quad g(1,i\bar{r}(\epsilon)) = -\frac{16}{27}  (1-2 \epsilon )^2 \left(\epsilon ^2-\epsilon +1\right)^2.
\end{equation}

We have,
\begin{equation}
f_{\rm MSD} = \frac{M^{(\pm)}(\epsilon)}{2N(\epsilon)},
\label{eq3.8}
\end{equation}
where
$N(\epsilon)$ in the denominator of Equation \eqref{eq3.8}, becomes
\begin{dmath}
    N(\epsilon)=S_1 -\frac{4}{9}(1-2 \epsilon )^2 ((\epsilon -1) \epsilon +1)S_2.
\end{dmath}
The form of the numerator $M^{(\pm)}$ depends on whether we have chosen to define our logical operators so that $M_3^{\otimes n}=\overline{M}_3$ or $M_3^{\otimes n}=\overline{M}_3^\dagger$. For $M_3^{\otimes n}=\overline{M}_3^\dagger$ the numerator of Equation \eqref{eq3.8} becomes, 
\begin{dmath}
M^{(-)}(\epsilon)= (2-2\epsilon)S_1 - \frac{8}{9} \epsilon ^2 (2 \epsilon-1) (\epsilon ^2-\epsilon +1) S_2.
\end{dmath}
For $M_3^{\otimes n}=\overline{M}_3$, the numerator of Equation \eqref{eq3.8} becomes 
\begin{dmath}
M^{(+)}(\epsilon)=2\epsilon S_1 + \frac{8}{9} (\epsilon -1)^2 (2 \epsilon -1) \left(\epsilon ^2-\epsilon +1\right) S_2.
\end{dmath}

As mentioned in section \ref{sec:M3-codes}, if we demand that the success probability for distillation be non-zero for pure input magic states, we should restrict attention to codes of length $n \equiv \pm 1 \pmod 6$, and choose logical operators so that
\begin{equation}
    M_3^{\otimes n}= \begin{cases}
        \overline{M}_3 & n \equiv 1 \pmod 6 \\
        \overline{M}_3^\dagger & n \equiv 5 \pmod 6.
    \end{cases} \label{M3-choice}
\end{equation} We will assume that this choice has been made in what follows, and refer to these as class-$5$ and class-$1$ codes respectively.

Our first result is the following theorem. 
\begin{theorem}
The noise suppression exponent of any magic state distillation routine based on an $[[n,1]]$ $M_3$-code, with logical operators chosen as per Equation \eqref{M3-choice}, and non-zero success probability for pure input magic states, must satisfy
\begin{equation}
\nu \equiv \begin{cases} 1 \pmod 3  &  n \equiv 1 \pmod 6 \\
2 \pmod 3   &  n \equiv 5 \pmod 6. \\
\end{cases}
\end{equation}
In particular, any such code with $n \equiv 5 \mod 6$ will distill $\ket{T}$-states with quadratic noise suppression. 
\label{theorem-mod-3}
\end{theorem}
\begin{proof}

It is convenient to define, \begin{equation}
    \phi(\epsilon)= -\frac{{ f}^3}{108  {g}}= \frac{(1-\epsilon )^3 \epsilon ^3}{(1-2 \epsilon )^2 \left(\epsilon ^2-\epsilon +1\right)^2}. \label{phi}
\end{equation}
We can rewrite $S_1$ and $S_2$ as sums over $\phi$
 as follows,
 \begin{equation}
    \tilde{S}_1=  \sum_{j=0}^{\floor{(n-1)/6}} c_{m-j} \phi^{j},  \quad \tilde{S}_2 = \sum_{j=0}^{\floor{(n-5)/6}}  d_{m-j} \phi^j, \label{sums}
\end{equation} with\footnote{We trust that the reader will not confuse $c_j$ defined here, for $[[n,1]]$ $M_3$-codes with the coefficients $c_j$ used in the section on $[[n,0]]$ $M_3$-states.}
\begin{eqnarray}
    c_j & = &  \left(-\frac{16}{27}\right)^j c'_j 4^{\frac{n-1}{2}-3 j}, \\
    d_j & = & \left(-\frac{16}{27}\right)^j d'_{j} 4^{\frac{n-5}{2}-3 j}.
\end{eqnarray}

Let us first consider codes with $n =6m+5$.  In terms of the new variables, $M^{(-)}$ takes the form,
\begin{equation}
    M^{(-)}= g(\epsilon)^m (2-2 \epsilon ) (1-\epsilon )^2 \epsilon ^2 \left(\tilde{S}_1+\frac{4}{9}H(\epsilon) \tilde{S}_2 \right),
\end{equation}
where 
\begin{equation}
H(\epsilon)=\frac{(2 \epsilon -1) \left(\epsilon ^2-\epsilon +1\right)}{(\epsilon -1)^3}. \label{H}
\end{equation}
For $\epsilon \in [0,1/2]$, we can write $H(\epsilon)$ as a function of $\phi$,
\begin{equation}
H(\epsilon(\phi))=\frac{\left(\sqrt{4 \phi +1}-1\right)}{2\phi }.
\end{equation}
Let us expand $H(\epsilon)$ as a power series in $\phi$, around $\phi=0$. We find
\begin{equation}
    H(\epsilon(\phi)) = \sum_{j=0}^\infty H_j \phi^j, \quad  H_j = \frac{(-4)^j \left(\frac{1}{2}\right)_j}{(2)_j}, \label{Hseries}
\end{equation}
where $\left(\frac{1}{2}\right)_{m-1}$ denotes the Pochhammer symbol: $a_n=a (a+1) \ldots  (a+n-1)=\frac{\Gamma  (a+n)}{\Gamma (a)}$.
The first few powers are given by,
\begin{equation}
     H(\epsilon(\phi)) \approx 1-\phi+2 \phi ^2 -5 \phi ^3 +14 \phi ^4 -42 \phi ^5 +O(\phi^6).
\end{equation}

We now expand $f_{\rm MSD}(\epsilon)=\frac{M^{(-)}(\epsilon)}{2N(\epsilon)}$ as a power series in $\epsilon$, and demand that $f_{\rm MSD} = O(\epsilon^\nu)$.
Because $N(\epsilon) \to 1$ by hypothesis, this implies,
\begin{equation}
\tilde{S}_1+\frac{4}{9}H(\epsilon(\phi)) \tilde{S}_2 = O(\phi^{\frac{\nu-2}{3}}). \label{extremal-5}
\end{equation}
Crucially, only non-negative integer powers of $\phi$ appear on the LHS of Equation \eqref{extremal-5}. This means that $\frac{\nu-2}{3}$ is a positive integer which proves our result. 

For $n \equiv 1 \pmod 6$, we have,
\begin{equation}
M^{(+)}(\epsilon)=2\epsilon g(\epsilon)^{m} H(\epsilon)^{-1}\left( H(\epsilon)\tilde S_{1}-\frac{4}{9}\tilde S_{2} \right).
\end{equation}
We now demand that $f_{\rm MSD}(\epsilon)=O(\epsilon^\nu),$
which implies that,
\begin{equation}
 H(\epsilon(\phi))\tilde{S}_1+ \tilde{S}_2 = O(\phi^{\frac{\nu-1}{3}}). \label{extremal-1}
\end{equation}
Again, $\frac{\nu-1}{3}$ must be a non-negative integer because $H(\epsilon(\phi))$ can be expanded as a power series in $\phi$ containing only non-negative integer powers of $\phi$. This proves the theorem for class 1 codes. 
\end{proof}

Note that, in the above analysis, we demand the denominator of Equation \eqref{eq3.8} is non-zero when $\epsilon \to 0$,
\begin{equation}
\lim_{\epsilon \to 0} N(\epsilon) \neq 0.
\end{equation} 
If this condition is not satisfied the code would be useless for magic state distillation. For class $5$ codes that meet this condition, we generically   expect quadratic noise suppression. However, if there are extra cancellations on the LHS of Equation \eqref{extremal-5}, we can get higher order noise suppression such as $\nu = 5$, $8$, $11$, or higher. From Figure \ref{fig:search-results}, we see that all $[[17,1]]$ $M_3$-codes give rise to quadratic noise suppression, and no such cancellation takes place. For class $1$ codes with $N(0) >0$, we generically expect  a linear relation between input and output error-rates; but, with cancellations on the LHS of \eqref{extremal-1}, we could obtain larger noise suppression exponents such as $\nu=4$, $7$, $10$ etc. 

Note that, by concatenating the 5-qubit code with itself $z$-times, we obtain a $[[5^z,1]]$ $M_3$-code, with $\nu=2^z$. One can check that, when $z$ is even, $5^z \equiv 1 \pmod 6$ and $2^z \equiv 1 \pmod 3$, and when $z$ is odd, $5^z \equiv 5 \pmod 6$ and $2^z \equiv 2 \pmod 3$ in accordance with Theorem \ref{theorem-mod-3}. For this family of codes $\frac{\nu}{n} \to 0$ as $z \to \infty$. Aside from the concatenated codes, such as these, no distillation protocols for $\ket{T}$ states with $\nu \geq 2$ are known.

\subsubsection{Extremal weight enumerators for distillation} 
We can extend the analysis used to prove Theorem \ref{theorem-mod-3} to prove the following bound on $\nu$ as a function of $n$, using ideas similar to those used to define extremal weight enumerators by Mallows and Sloane  \cite{MALLOWS1973188}. We are interested in the noise-suppression exponent rather than the distance; in analogy to \cite{MALLOWS1973188}, we define an \textit{extremal weight enumerator for distillation}, to be one with the largest possible noise suppression exponent $\nu_{\rm max}$ consistent with invariant theory. We then show that, for any such enumerator, one of the $A_w$ is negative, which implies $\nu \leq \nu_{\rm max}$ for any realizable code.

\begin{theorem}
   For any distillation routine based on an $M_3$-code, 
   \begin{equation}
   \nu \leq n-3.
   \end{equation} \label{thm:extremal-no-go}
\end{theorem}
\begin{proof}
We first ignore both the classical and quantum constraints, and simply solve for $c_j$ and $d_j$ so that the LHS of \eqref{extremal-1} and \eqref{extremal-5} correspond to a value of $\nu$ that is as large as possible.

For class-$5$ codes, there are $2m+1$ unknown coefficients: $c_j$, for $j=1, \ldots, m$ and $d_k$, for $k=0, \ldots m$. By using these $2m+1$ unknowns to cancel the first $2m+1$ powers of $\phi$ that appear in the LHS of Equation \eqref{extremal-5}  -- we obtain $\frac{\nu_{\rm max}-2}{3}=2m+1$ for an extremal distillation routine for class 5 codes. Similar analysis can also be carried out for class 1 codes, for which we find $\frac{\nu_{\rm max}-1}{3}=2m$. The weight enumerators determined by this procedure are unique, and we refer to them as ``extremal enumerators for distillation''.

These results imply that $\nu_{\rm max}=n$, which, in contrast with the case of extremal self-dual codes, is a trivial upper bound. However, we can show that no codes with $\nu=\nu_{\rm max}$ exist. 

The extremal weight enumerators for distillation for $n<20$ that arise from this procedure are given in Table \ref{tab:extremal-msd-enumerators}. As can be seen from the table, we find that many coefficients of $A$ are negative. In general, one can show that $A_2$ is always negative:
\begin{equation}
    A_2 = \begin{cases} -30 - 81 m - 54 m^2 & \text{ for } n=6m+5, \\
    -9m - 54 m^2 & \text{ for } n =6m+1.
    \end{cases}
\end{equation}
 Therefore codes realizing extremal weight enumerators for distillation do not exist\footnote{Interestingly, however, the $A_w$ are all integers divisible by $3$ for $w >0$.}, and (using Theorem \ref{theorem-mod-3}) we have the result that $\nu \leq n-3$, as claimed. 
\end{proof}

In the next subsection, we use linear programming to place stronger bounds on $\nu$.

\begin{table}[ht]
    \centering
    \renewcommand{\arraystretch}{1.4} 
    \setlength{\tabcolsep}{8pt}      
    \begin{tabular}{@{}c p{0.75\linewidth}@{}}
        \toprule
        \textbf{$n$} & \textbf{$A(1, y)$} \\ 
        \midrule
        5  & $1 - 30y^2 + 45y^4$ \\
        7  & $1 - 63y^2 + 315y^4 - 189y^6$ \\
        11 & $1 - 165y^2 + 2970y^4 - 12474y^6 + 13365y^8 - 2673y^{10}$ \\
        13 & $1 - 234y^2 + 6435y^4 - 46332y^6 + 104247y^8 - 69498y^{10} + 9477y^{12}$ \\
        17 & $1 - 408y^2 + 21420y^4 - 334152y^6 + 1969110y^8 - 4725864y^{10} + 4511052y^{12} - 1487160y^{14} + 111537y^{16}$ \\
        19 & $1 - 513y^2 + 34884y^4 - 732564y^6 + 6122142y^8 - 22447854y^{10} + 36732852y^{12} - 25430436y^{14} + 6357609y^{16} - 373977y^{18}$ \\
        \bottomrule
    \end{tabular}
    \caption{Extremal weight enumerators for distillation for $n < 20$.}
    \label{tab:extremal-msd-enumerators}
\end{table}

\subsection{Linear programming bounds for the noise suppression exponent}

In this section, we seek to strengthen Theorem \ref{thm:extremal-no-go} to determine the largest value of $\nu$ attainable by weight enumerators consistent with \textit{all} classical and quantum constraints on the coefficients $c_i$ and $d_i$. This is essentially a linear programming problem, although the quantum constraints are not quite linear. To arrive at our main result, Theorem \ref{conjecture-1} below, we essentially solve this linear programming problem using a computer algebra system\footnote{We used the package Mathematica 12, and should emphasize that we did not rely on floating point arithmetic. Quantum constraints are difficult to implement for larger $n$ -- we only attempted to demand consistency with quantum constraints for $n<50$; for larger values of $n$, we only demanded that $A(1, \frac{i}{\sqrt{3}})\geq 0$.} Explicit examples illustrating exactly how this is done are given in Section \ref{sec:11-1} and Appendix \ref{app:small-examples}. 

\begin{theorem}
For all $n\leq 215$, the noise suppression exponent $\nu$ of a magic state distillation routine based on a non-trivial $M_3$-code is bounded from above as follows.

For codes of length $n=6m+1$,
\begin{equation}
    \nu \leq 3m-5.
\end{equation}
For codes of length $n=6m+5$,
\begin{equation}
    \nu \leq \begin{cases}
    2  & \text{    for }n=5,~11 \\3m-4 & \text{    for } n \geq 17.
    \end{cases}
\end{equation}
\label{conjecture-1}
\end{theorem}
Although we have only proven Theorem \ref{conjecture-1} via computer-algebra for $n \leq 215$, we assert that it is almost certainly true for all $n$ -- in principle, techniques such as those used by Rains \cite{rains} could be used to provide an analytical proof, although the combinatorial expressions involved may be rather tedious. 

Our computer search shows that the bounds of theorem \ref{conjecture-1} are saturated for $n \leq 17$. However, the bounds are not saturated for $n=19$; we found no code with noise suppression coefficient greater than $1$, while the Theorem \ref{conjecture-1} allows $\nu=4$. For $19 \leq n\leq 35$, we are able to find integral weight enumerators saturating the bounds of Theorem \ref{conjecture-1} that satisfy all quantum and classical constraints. Some of these are presented in Appendix \ref{app:integer-programming}. 

\section{Discussion}
\label{discussion}
In this work, we demonstrated that the physical constraints of quantum information processing impose rigorous new bounds on classical coding theory. We showed that the performance of a magic state distillation \cite{MSD} protocol based on a linear $GF(4)$ code is captured by its classical simple weight enumerator. This allows us to place new constraints on weight enumerators, which arise from ``quantum consistency'' of the magic state distillation routine. These new constraints allowed us to rule out the entire family of extremal type $4^H$ self-dual codes of length $12m$ -- a previously unresolved problem in classical coding theory \cite{macwilliams1978self, lam1990there, huffman1990extremal, huffman1991extremal}.
We were also able to use linear programming techniques to place bounds on the best attainable distances and to restrict the best attainable noise suppression exponents of an $[[n,1]]$ $M_3$-codes. 

Using these results, we carried out an exhaustive search of distillation routines via $M_3$-codes of size $n<20$. While we were not able to find a distillation routine with a threshold exceeding that of the 5-qubit code, we hope to extend this computational search to larger codes in the near future. Notably, we were able to construct many weight enumerators with integer coefficients that satisfy all known classical and quantum consistency conditions for codes of size $n\geq23$. If realizable, these enumerators would correspond to magic state distillation routines with better thresholds than the 5-qubit code. It remains a fascinating open question whether a physical quantum code realizing any of these weight enumerators actually exists.

Our findings also open several new avenues in classical coding theory and related fields. While we have identified two new quantum consistency constraints, it is unknown whether additional quantum constraints exist. Furthermore, is there a purely classical combinatorial interpretation of $A(1,\frac{i}{\sqrt{3}})$ that explains why it cannot be negative? Mallows and Sloane  \cite{mallows1981weight},showed that complete weight enumerators can provide stronger constraints on the existence of certain self-dual codes than simple weight enumerators; extending our analysis to complete weight enumerators may yield even tighter bounds on quantum codes. It would also be of great interest to extend these quantum constraints to place new bounds on Type II self-dual binary codes, and to investigate possible applications to the modular bootstrap program \cite{dymarsky2021solutions, dymarsky2021quantum, dymarsky2023fake, angelinos2022optimal, alam2023narain}.

Finally, from a quantum information perspective, we restricted our attention in this paper to distillation routines for the $\ket{T}$ magic state of Bravyi and Kitaev \cite{MSD}, which are at present the most mysterious and poorly understood. It would be valuable to apply these methods to study distillation routines for $\ket{H}$-type magic states. Furthermore, extending this framework beyond qubits to qudits of odd dimension presents a profound theoretical opportunity. For odd-prime qudits, non-stabilizerness has a precise operational meaning as contextuality with respect to stabilizer measurements \cite{nature}. Determining the best-attainable threshold for qudit magic state distillation would directly address whether contextuality is the essential feature of quantum mechanics responsible for the power of quantum computers. (See also \cite{PhysRevA.95.052334,PhysRevLett.122.140405, DPS2, Delfosse_2017}.) It would therefore be worthwhile to apply invariant theory to constrain magic state distillation routines for qudit states --particularly the qutrit strange state, which, as shown by Prakash and Singhal  \cite{prakash2024search}, is also characterized by an unsigned simple weight enumerator. 

\section*{Acknowledgments}
SP thanks Prof. P.S. Satsangi for inspiration and guidance. Both SP and ARK thank Eric Rains for several helpful comments and suggestions. We also thank Markus Grassl for pointing out typos in an earlier draft of this paper. ARK would like to thank Debbie Leung, Michele Mosca and Graeme Smith for useful discussions.

SP acknowledges the support of MeitY QCAL, Amazon Braket and DST-SERB grant (CRG/2021/009137). SP also thanks the International Centre for Theoretical Sciences (ICTS), Bengaluru for hospitality during the program ``A Hundred Years of Quantum Mechanics" (ICTS/qm100-2025/01) where part of this work was completed.
ARK acknowledges NTT research for financial support.  Research at Perimeter Institute and IQC is supported by the Government of Canada through Innovation, Science and Economic Development Canada, and by the Province of Ontario through the Ministry of Research, Innovation and Science.

\appendix

\section{Linear programming bounds for small $n$}
\label{app:small-examples}
In this appendix we explicitly construct the space of weight enumerators of  $[[n,0]]$ $M_3$-codes and $[[n,1]]$ $M_3$-codes consistent with invariant theory and the classical and quantum linear programming bounds,  for small values of $n$. The examples in this appendix serve two purposes: they illustrate how the quantum constraints drastically reduce the space of consistent weight enumerators, and they illustrate how Theorem \ref{conjecture-1} was obtained. 

\subsection{$[[n,0]]$ $M_3$-states}

We will always fix $c_0$ by demanding $A_1=1$. Then for $n<6$, there is no non-trivial weight enumerator. 
\subsubsection*{Bounds for $6 \leq n \leq 10$}
For $n = 6$, $8$, and $10$, the weight enumerator depends on a single undetermined coefficient $c_1$. The classical bounds, $A_i\geq 0$ translate to:
\begin{eqnarray}
    n&=6:  \quad & -9 \leq c_1 \leq \frac{27}{2},  \\
    n&=8:  \quad& -12 \leq c_1 \leq \frac{108}{5}, \\
    n&=10: \quad & -15 \leq c_1 \leq \frac{135}{4}. 
\end{eqnarray}
The quantum bound reduces to simply $c_1 \leq 0$. 
These are plotted, along with points representing weight enumerators of all codes that exist (from Munemasa's database \cite{munemasa_codes_website}), in Figure \ref{fig:small-code-bounds}. We see that, when the quantum bound is included, both sides of the linear programming bounds are saturated.

\begin{figure}
   
    \hspace{2cm}\includegraphics[width=0.53\linewidth]
    {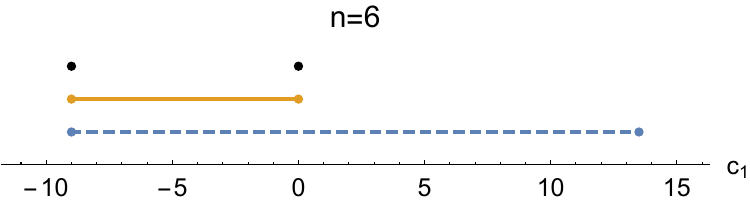}\vspace{0.5cm}
    
    \hspace{2cm}\includegraphics[width=0.53\linewidth]{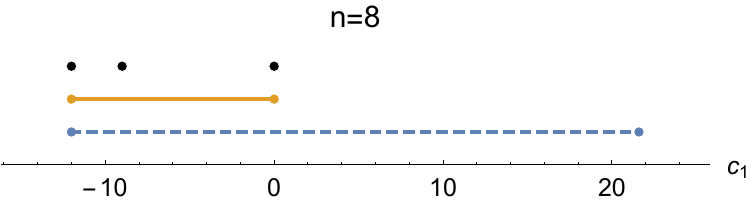}
    \vspace{0.5cm}
    
    \hspace{2cm}\includegraphics[width=0.8\linewidth]{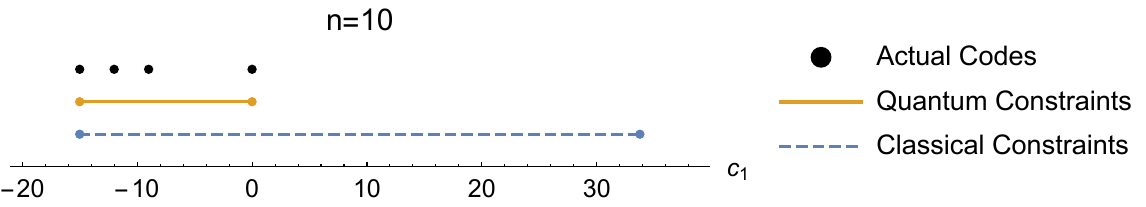}
    \caption{Weight enumerators for self-dual linear $[n,n/2]_{GF(4)}$codes with $n=6$, $8$ and $10$ are parameterized by a single real number $c_1$. The range of $c_1$ allowed by classical constraints, $A_i\geq0$, is shown as a blue dashed line, and the range allowed by quantum constraints is shown as an orange solid line. Weight enumerators of all self-dual codes that exist are shown as black points. }
    \label{fig:small-code-bounds}
\end{figure}

\subsubsection*{Bounds for $12 \leq n \leq 16$}
For $n=12$, $14$ and $16$ a general weight enumerator is determined by two coefficients $c_1$ and $c_2$. The classical bounds for $n=14$ and $16$ are derived in the same manner as for $n=12$ in the main text. The quantum bound remains $1-c_1 \phi+c_2 \phi^2 \geq 0$ (for $\phi \in [0,\infty)$), which again reduces to
\begin{equation}
    c_2 \geq \begin{cases}
         0 & c_1 <0 \\
         c_1^2/4 & c_1 \geq 0,
    \end{cases}
\end{equation}
as for $n=12$. Linear programming bounds for $n=14$ and $n=16$ are shown in Figure \ref{fig:14-16-codes}.

\begin{figure}
    \centering
    \includegraphics[height=0.40
\linewidth]{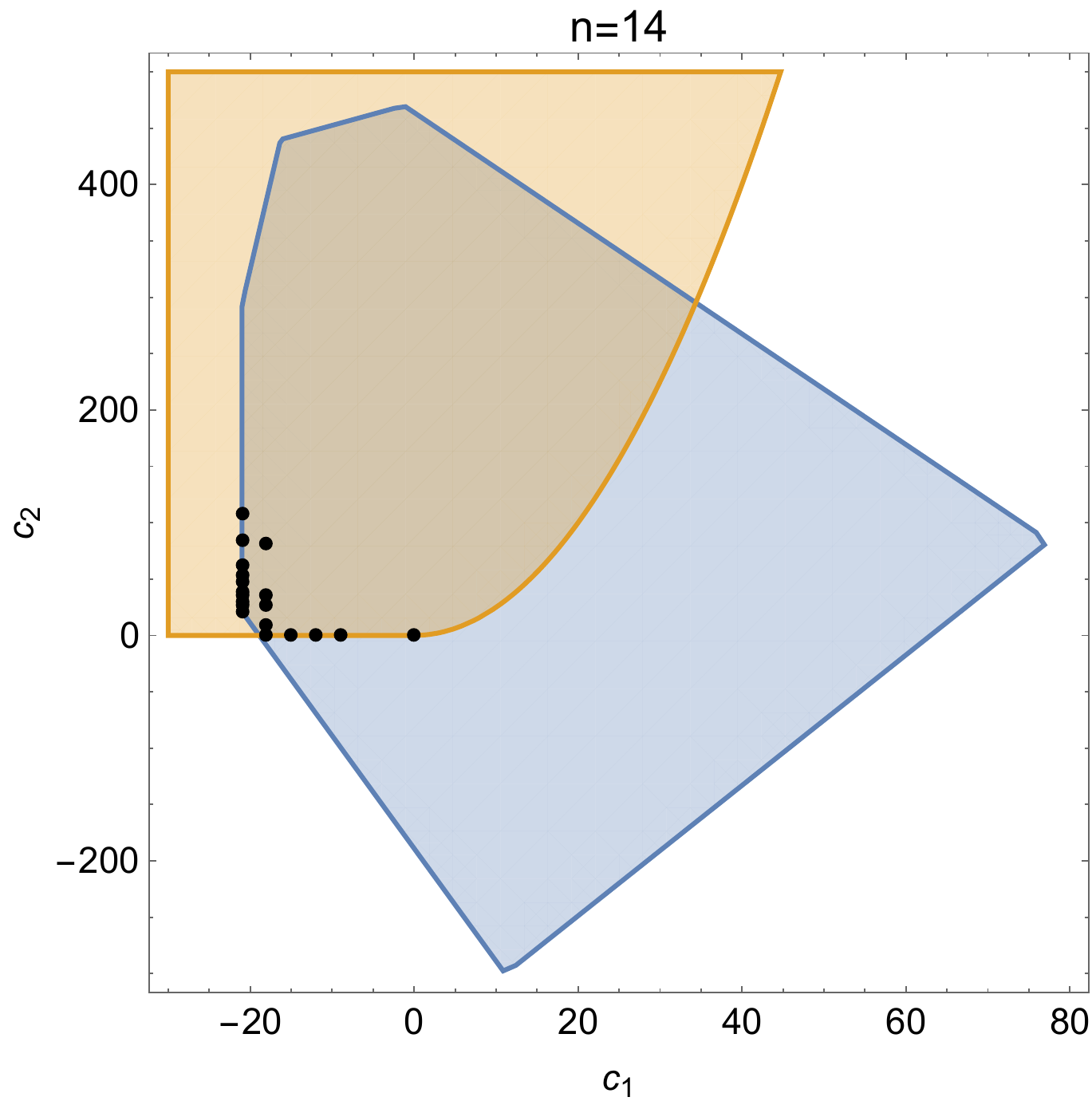}
    \includegraphics[height=0.42\linewidth]{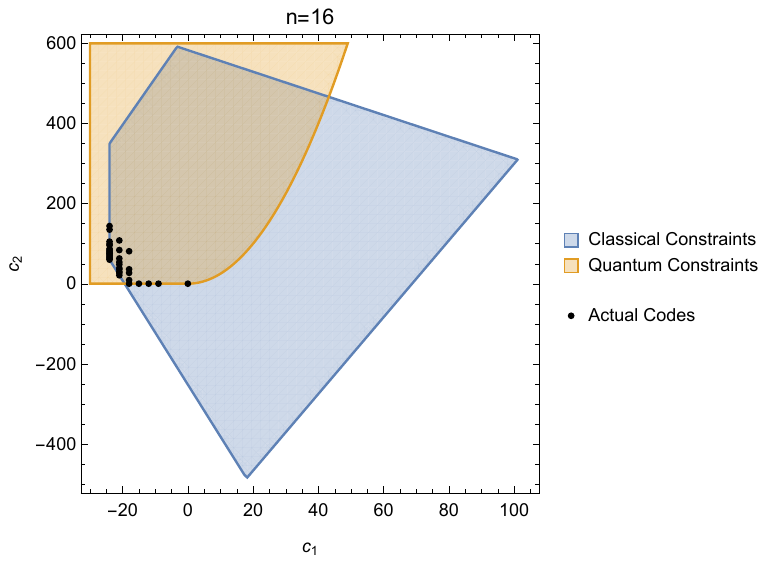}
    \caption{Classical (blue) and quantum (orange) linear programming bounds for classical self-dual linear  $[14,7]_{GF(4)}$ (left) and $[16,8]_{GF(4)}$ (right) codes. All self-dual codes of these lengths which exist are shown as black points.}
    \label{fig:14-16-codes}
\end{figure}

\subsection{$[[n,1]]$ $M_3$-codes}
Here we illustrate the role played by the new quantum constraints in constraining the space of consistent weight enumerators for small values of $n$. 

\subsubsection*{Illustration for $n=5$}
\label{sec:n5}
Let us constrain the space of allowed $[[5,1]]$ $M_3$-codes. We will use the notation of the proof of Theorem \ref{theorem-mod-3}.  
After demanding $A_0=1$, a general weight enumerator depends on only one undetermined coefficient $d_{0}$: 
\begin{equation}
A(x,y)=x^{5}+(6+d_{0})x^3y^2+(9-d_{0})xy^{4}.
\end{equation}
Equation \eqref{extremal-5}, which determines the  noise-suppression-exponent $\nu$, becomes, 
\begin{eqnarray}
    \tilde{S}_1+H(\epsilon(\phi)) \tilde{S}_2 =
    16 + \frac{4}{9}(1-\phi + \ldots) d_0 = O(\phi^{\frac{\nu-2}{3}}). 
\end{eqnarray}
We choose ${d}_{0}$ so that the order $\phi^0$ term vanishes:
\begin{equation}
d_{0}=-36,
\end{equation}
to obtain a noise suppression exponent,
\begin{equation}
\frac{\nu-2}{3}=1 \implies \nu=5.
\end{equation}
This condition defines the extremal weight enumerator for distillation. Substituting into the expression for $A(x,y)$ we find,
\begin{equation}
A(x,y)=x^{5}-30x^{3}y^{2}+45xy^{4}.
\end{equation}
The coefficient of $x^{3}y^{2}$ is negative. Therefore, the extremal weight enumerator for distillation cannot be realized, and we have the bound that $\nu \leq 2$.

We now leave $d_0$ a free parameter, and determine its allowed range of $d_0$ via classical and quantum constraints. We write $B(x,y)$ as a function of $d_{0}$ and find,
\begin{equation}
B(x,y)=x^{5}+\left(3+\frac{d_{0}}{2}\right)x^{4}y+\left(6+d_{0}\right)x^3y^2+\left(18-2d_{0}\right)x^{2}y^{3}+\left(9-d_{0}\right)xy^{4}+\left(27+\frac{3}{2}d_{0}\right)y^{5}.
\end{equation}
We find that the conditions $\{B_w \geq 0\}$ for all $w$ reduce to  
\begin{equation}
-6 \leq d_0 \leq 9.
\label{5-qubit-eqconstraint1}
\end{equation}
This is the solution to the classical constraints on weight enumerators.

We now demand that the success probability, $2^{-n+1}W_I(\bar{r})=2^{-n+1} A(1,ir)=2^{-n} N(\epsilon)$, be non-negative for $-\frac{1}{3} \leq \bar{r}^2 \leq \frac{1}{3}$. We find,
\begin{equation}
A(1,ir)=1-(6+d_{0})\bar{r}^2+(9-d_{0})\bar{r}^4.
\end{equation}
At $\bar{r}=\frac{1}{3}$, this evaluates to $\frac{-4d_{0}}{9}$, which implies $d_{0}\leq 0$. Combining the quantum and classical constraints we therefore have, 
\begin{equation}
-6 \leq d_0 \leq 0.
\label{5-qubit-eqconstraint2}
\end{equation}
It turns out both sides of this inequality are saturated by 5-qubit $M_3$-codes. The weight enumerator of the 5-qubit code is obtained by taking $d_0=-6$. There is one other 5-qubit $M_3$-code, whose weight enumerator is obtained by setting $d_0=0$. (This other code has zero probability of successfully projection on pure $\ket{T}$ states, so it is useless for magic state distillation.) Figure \ref{fig:5-linear-programming} shows the region allowed by classical constraints, the quantum constraints, and the two actual weight enumerators of $[[5,1]]$ $M_3$-codes.

\begin{figure}
    \centering
    \includegraphics[width=0.8\linewidth]{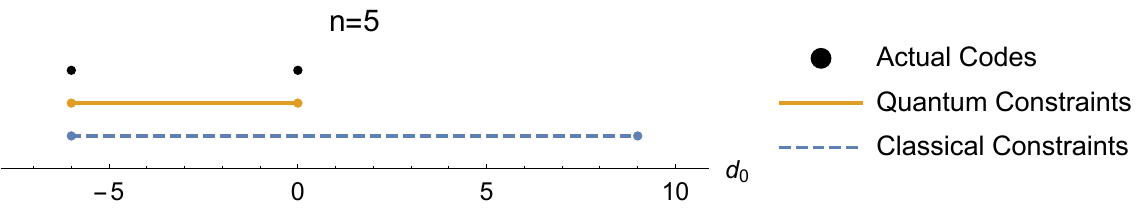}
    \caption{Linear programming bounds for $n=5$ weight enumerators. The dashed blue line corresponds to the classical linear programming constraints (Equation \ref{5-qubit-eqconstraint1}) and the yellow line corresponds to the bound from the quantum success probability constraint (Equation \ref{5-qubit-eqconstraint2}). {  The} black  points are weight enumerators for the two inequivalent $[[5,1]]$ $M_3$-codes, which saturate the quantum constraints. 
    }
    \label{fig:5-linear-programming}
\end{figure}

\subsubsection{$n=7$}
The weight enumerator for a $[[7,1]]$ $M_3$-code depends on two undetermined coefficients, $c'_1$ and $d'_0$, and is
\begin{equation}
    A(1, y) = 1 
            + y^2 \left( c_1' + d_0' + 9 \right) 
            + y^4 \left( -2c_1' + 2d_0' + 27 \right) 
            + y^6 \left( c_1' - 3d_0' + 27 \right).
\end{equation}
The classical linear programming constraints, $B_w \geq 0$  reduce to,
\begin{align*}
    d'_0  & \geq -6, \\
    c'_1 + d'_0 & \geq -9, \\
    6c'_1 -d'_0 & \geq -54, \\
    c'_1 -d'_0  & \leq \frac{27}{2}, \\
    4c'_1 + 3d'_0 & \leq 54, \\
    c'_1 - 3d'_0 & \geq -27 . 
\end{align*}
The convex polytope defined by these inequalities is shown as the blue hexagon in Figure \ref{fig:7-1-linear-programming}.

Let us now look at the additional quantum constraints that arise from viewing the code as a potential magic state distillation routine. The constraint that $N(0)\geq 0$ translates into $c_1'\leq 0$. This is a special case of the infinite family of constraints that $N(\epsilon) \geq 0$ $\forall\epsilon \in [0,1]$, which reduce to a non-linear constraint on $c_1'$ and $d_0'$, shown in green in Figure \ref{fig:7-1-linear-programming}. The distillation threshold constraint depends on the choice of logical operators, and holds for both choices. The region allowed by both these constraints is shown as the orange region in Figure \ref{fig:7-1-linear-programming}, which is non-convex. 

We see from Figure \ref{fig:7-1-linear-programming} that both quantum constraints are independent of each other, as well as the classical constraints. Moreover, the four $[[7,1]]$ $M_3$-codes that exist lie on the extremal points of the intersection of the quantum and classical constraints. 

Let us illustrate how the bound on noise suppression exponent from Theorem \ref{conjecture-1} arises for this case. The line corresponding to enumerators with noise suppression exponent $\nu \geq 4$ corresponds to the equation $c'_1=-3d'_0$, which is shown in gray in Figure \ref{fig:7-1-linear-programming}. It only intersects the allowed region at the point one $(d_0',c_1')=(0,0)$, which has zero success probability at $\epsilon=0$ and is therefore useless for magic state distillation. In particular, it intersects the line of non-trivial codes, specified by $C_1=0 \implies d_0'=-6$, far outside the blue region, hence demonstrating that $\nu \leq 1$ purely from linear programming.

\begin{figure}
    \centering
    \includegraphics[width=0.95\linewidth]{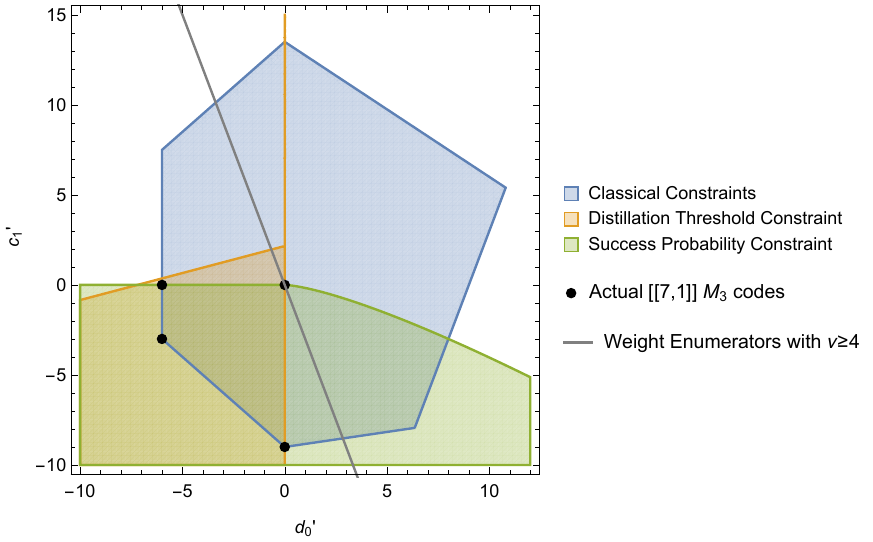}
    \caption{Weight enumerators for $[[7,1]]$ $M_3$-codes are characterized by two real parameters $d_0'$ and $c_1'$. The blue region is defined by the classical constraints $B_i \geq 0$. The orange region is defined by the quantum constraint that the threshold lie outside the stabilizer octahedron, $\epsilon_* \leq \epsilon_{\rm max}$, and the green region is defined by  the quantum constraint that the success probability be non-negative, $N(\epsilon)\geq 0$, for all $\epsilon<\epsilon_{\rm max}$.  Black points denote all $[[7,1]]$ $M_3$-codes, which lie at the boundary of the intersection of all three regions. The line corresponding to enumerators with noise suppression exponent $\nu \geq 4$ is shown in gray -- it only intersects one enumerator which has zero success probability.}
    \label{fig:7-1-linear-programming}
\end{figure}

It is also possible to solve the integer linear programming problem exactly for this case. One can check that, for $c'_1 \equiv 0 \pmod 3$ and $d'_0 \equiv 0 \pmod 6$, solutions to the linear programming problem become solutions to integer programming, where we have the additional constraint that $B_i$ are integers  such that $B_i \equiv 0 \pmod 3$.  Without quantum constraints, there are $18$ integer solutions. This is reduced to 6 integer solutions when the two quantum constraints are included.\footnote{The two putative weight enumerators that do not correspond to codes are defined by $(d_0',c_1')=(0,-3)$ and $(0,-6)$; they would both correspond to trivial distillation routines which would have $\epsilon_{\rm out}=\epsilon$.}

\subsubsection{$n=13$}
The case of $n=11$ was covered in the main text. We now turn to $n=13$. A general weight enumerator for $n=13$ depends on four coefficients, but we can eliminate $d_0'$ and $c_1'$ by demanding $C_1=0$ and $A_2=0$. A plot of all $[[13,1]]$ codes meeting these constraints is shown in Figure \ref{fig:13-linear-programming}. The proof of Theorem \ref{conjecture-1} for this value of $n$ follows from observing that the line determined by demanding $\nu \geq 4$ never intersects the allowed region. 

\begin{figure}
    \centering
    \includegraphics[height=0.5\linewidth]{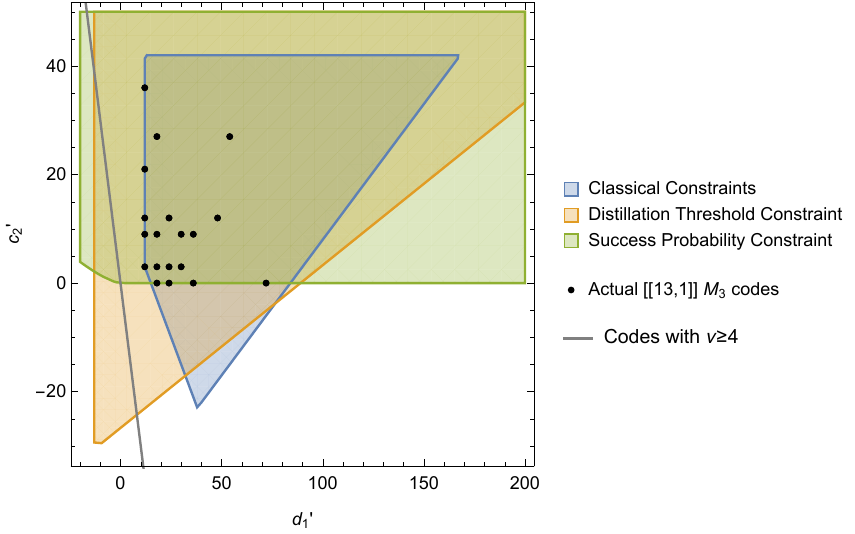}
    \caption{Linear programming bounds for $n=13$, in the plane defined by $C_1=A_2=0$. Regions are labeled as in Figure \ref{fig:7-1-linear-programming}. The gray line denotes those enumerators for which $\nu \geq 4$, and never intersects the allowed region.}
    \label{fig:13-linear-programming}
\end{figure}

\subsubsection{$n=17$}
\label{sec:17-1}
A general weight enumerator for $n=13$ depends on 5 coefficients. We can eliminate $d_0'=-6$ and $c_1'=-18$ by demanding $C_1=0$ and $A_2=0$. We can eliminate one more variable by demanding either $\nu \geq 5$ or $C_3=0$. 

We find that, when restricting to the $\nu \geq 5$ region, no weight enumerator satisfies the threshold constraint for all $\epsilon_{\rm in} \in [\epsilon_{\rm max}, 1/2]$, so $\nu=5$ is ruled out. Plots of the allowed regions for enumerators  with all $[[17,1,d]]$ $M_3$-codes with $d>3$ are shown in Figure \ref{fig:17-linear-programming-d}.

\begin{figure}
    \centering
    \includegraphics[height=0.5\linewidth]{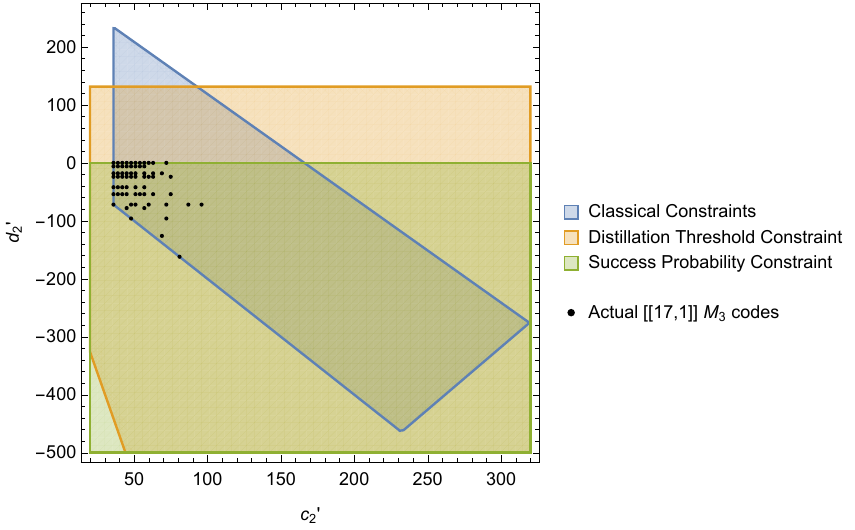}
    \caption{Linear programming bounds for $n=17$, in the plane defined by $C_1=A_2=0$ and $d \geq 5$. Regions are labeled as in Figure \ref{fig:7-1-linear-programming}, and actual codes are shown as black points. }
    \label{fig:17-linear-programming-d}
\end{figure}

\subsection{More linear programming bounds}
\label{app:distance}

We solved the linear programming problem for $[[n,1]]$ $M_3$-codes as in the previous subsection for larger values of $n$ on a computer algebra system. We find that one can obtain stronger bounds on the distance of $[[n,1]]$ $M_3$-codes than given in Theorem \ref{thm:distance-M_3-code}. These are:
\begin{theorem}
The distance of an $[[n,1,d]]$ $M_{3}$-code satisfies $d<d_{\max}$, with  
\begin{equation}
    d_{\text{max}} \leq \begin{cases}
        2m-1  & n = 6m+1 \in [121, 229 ], \quad
        n=6m+3 \in [135, 243 ]\\
        2m+1  & n = 6m+5 \in [101,209].
    \end{cases}
\end{equation}
for $101\leq n \leq 209$.
\end{theorem}
For larger $n$ we have the following conjecture which is consistent with computational evidence,
\begin{conjecture}
The distance of an $[[n,1,d]]$ $M_{3}$-code satisfies $d<d_{\max}$, with  
\begin{equation}
    d_{\text{max}} \leq \begin{cases}
        2m-3  & n = 6m+1\geq 235, ~
        n=6m+3 \geq  249 \\
        2m-1  & n = 6m+5\geq 215.
    \end{cases}
\end{equation}
\end{conjecture}

\section{Integer programming}
\label{app:integer-programming}

In this appendix, we attempt to generate weight enumerators for putative codes with high thresholds and noise suppression exponents with integer coefficients that meet all classical and quantum consistency conditions. Recall that, if an $M_3$ code contains a stabilizer $P$, it also contains $P'=M_3^{\otimes n} P (M_3^\dagger)^{\otimes n}$ and $P''=M_3^{\otimes n} P (M_3^\dagger)^{\otimes n}$. Because $\wt(P)=\wt(P')=\wt(P'')$, the number of stabilizers of non-zero weight must be divisible by three.  

Solving for integral weight enumerators is essentially an intractable problem. However, we found an effective strategy that works in practice. We first choose to treat the $B_w$ as independent variables rather than $c_j$ and $d_j$. We choose to demand a noise suppression exponent $\nu=2m+1-L$. This leaves us with $L$ undetermined variables:  $\{B_1, \ldots, B_L\}$. We demand that these are integers divisible by three, and try to find a solution to quantum and classical constraints by brute force. The other $B_w$ for $w>L$ are determined in terms of $\{B_1, \ldots, B_L\}$. Usually, when we find a solution, all the other $B_w$ for $w>L$ also turn out to be integers divisible by three.

We list some putative weight enumerators for $n \leq 35$, that pass all our consistency constraints, as well as the non-triviality constraint $B_1=0$ and $B_2=0$, with large noise suppression exponent. Of course, most solutions to integer linear programming constraints, including our new quantum constraints, do not correspond to codes. However, we tried to construct examples with large distance and low degeneracy, that may be somewhat more likely to exist than generic weight enumerators. 

For $n=19$, we find several weight enumerators with $\nu=4$, such as, 
\begin{equation}
\begin{aligned}
    A(1, y) = & \, 1 + 36 y^6 + 1194 y^8 + 9108 y^{10} + 53736 y^{12} + 103404 y^{14} + 
 80877 y^{16} + 13788 y^{18},
\end{aligned}
\end{equation}
which would give rise to a code with $\epsilon_{\rm out} \approx  395 \epsilon^4 +O(\epsilon^5)$. However, none of the weight enumerators saturating Theorem \ref{conjecture-1} correspond to actual codes.

For $n \geq 23$, we can find many putative weight enumerators with $\nu \geq 5$  with integer coefficients satisfying all our constraints. A complete classification of $M_3$-codes of size $n \geq 23$ does not exist in the literature, so we do not yet know if codes possessing these weight enumerators exist.

We first list some class 5 enumerators, with $\nu \geq 5$:
\begin{itemize}
    \item A weight enumerator for a putative $[[23,1,7]]$ code with $\epsilon_{\rm out} \approx  587 \epsilon ^5 +O(\epsilon^6)$, and threshold $\epsilon_*=0.175343$:
 \begin{equation}
    \begin{split}
    A(1, y) = & 1+90 y^6+1314 y^8+348 y^{10}+107280 y^{12} \\ & +434880
   y^{14}+1282869 y^{16}+1543428 y^{18}+738072
   y^{20}+86022 y^{22}.
    \end{split}
\end{equation}
    This weight enumerator saturates both the bound on $\nu$ in Theorem \ref{conjecture-1} and the bound on distance in Theorem \ref{thm:distance-bound}, so it would be particularly interesting to know whether or not a corresponding code exists.

 \item A weight enumerator for a putative $[[29,1,7]]$ code with $\epsilon_{\rm out} \approx  \frac{52999 \epsilon ^8}{2} +O(\epsilon^9)$, and threshold $\epsilon_*=0.211288$:
\begin{equation}
\begin{aligned}
    A(1, y) = & \, 1 + 810y^6 + 8985y^8 + 13134y^{10} + 19728y^{12} + 362820y^{14} \\
              & + 6283203y^{16} + 24574140y^{18} + 64556616y^{20} + 91398066y^{22} \\
              & + 63125091y^{24} + 16705494y^{26} + 1387368y^{28}.
\end{aligned}
\end{equation}

 \item A weight enumerator for a putative $[[35,1,7]]$ code with $\epsilon_{\rm out} \approx  1116496 \epsilon ^{11} +O(\epsilon^{12})$, and threshold $\epsilon_*=0.21123$:
\begin{equation}
\begin{aligned}
    A(1, y) = & \, 1 + 2988y^6 + 45834y^8 + 296382y^{10} + 939012y^{12} + 390696y^{14} \\
              & + 478032y^{16} + 36246354y^{18} + 352884120y^{20} + 1382651340y^{22} \\
              & + 3406283646y^{24} + 5132651490y^{26} + 4553571492y^{28} + 1995562656y^{30} \\
              & + 279328743y^{32} + 38536398y^{34}.
\end{aligned}
\end{equation}

\item A weight enumerator for a putative $[[35,1,11]]$ code with $\epsilon_{\rm out} \approx  \frac{11781 \epsilon ^5}{23} +O(\epsilon^6)$, and threshold $\epsilon_*=0.16331$:
  \begin{equation}
    \begin{split}
    A(1, y) = & \ 1 + 42840y^{12} + 6715170y^{16} + 46236960y^{18} + 339481296y^{20} \\
              & + 1334551680y^{22} + 3443179320y^{24} + 5213799360y^{26} \\
              & + 4481873880y^{28} + 1943770752y^{30} + 353253285y^{32} \\
              & + 16964640y^{34}.
    \end{split}
\end{equation}

\end{itemize}

Some potential class $1$ weight-enumerators with $\nu \geq 7$ are:
\begin{itemize}
    \item A weight enumerator for a putative $[[25,1,3]]$ code with $\epsilon_{\rm out} \approx  \frac{23591 \epsilon ^7}{5} +O(\epsilon^8)$, and threshold $\epsilon_*=0.209325 $:
\begin{equation}
\begin{aligned}
    A(1, y) = & \, 1 + 39y^4 + 1155y^6 + 8679y^8 + 8796y^{10} + 112482y^{12} \\
              & + 487338y^{14} + 2805963y^{16} + 5398860y^{18} + 5548959y^{20} \\
              & + 2268459y^{22} + 136485y^{24}.
\end{aligned}
\end{equation}

       \item A weight enumerator for a putative $[[25,1,5]]$ code with $\epsilon_{\rm out} \approx  21181 \epsilon ^7 +O(\epsilon^8)$, and threshold $\epsilon_*=0.18769  $:
\begin{equation}
\begin{aligned}
    A(1, y) = & \, 1 + 21y^4 + 759y^6 + 5109y^8 + 8556y^{10} + 61854y^{12} \\
              & + 578946y^{14} + 2755767y^{16} + 5546652y^{18} + 5370165y^{20} \\
              & + 2263167y^{22} + 186219y^{24}.
\end{aligned}
\end{equation}

    \item A weight enumerator for a putative $[[31,1,3]]$ code with $\epsilon_{\rm out} \approx  {1144115 \epsilon ^{10}} +O(\epsilon^{11})$, and threshold $\epsilon_*=0.210862$:
\begin{equation}
\begin{aligned}
    A(1, y) = & \, 1 + 522y^4 + 15045y^6 + 163953y^8 + 889332y^{10} + 2605962y^{12} \\
              & + 3622425y^{14} + 3506211y^{16} + 24209232y^{18} + 130869318y^{20} \\
              & + 302765415y^{22} + 355566699y^{24} + 209428092y^{26} + 39798054y^{28} \\
              & + 301563y^{30}.
\end{aligned}
\end{equation}
    
    \item A weight enumerator for a putative $[[31,1,9]]$ code with $\epsilon_{\rm out} \approx  {18569 \epsilon ^7} +O(\epsilon^8)$, and threshold $\epsilon_*=0.174321 $:
 \begin{equation}
\begin{aligned}
    A(1, y) = & \, 1 + 369y^6 + 2898y^8 + 4521y^{10} + 57951y^{12} + 466488y^{14} \\
              & + 6245181y^{16} + 36350466y^{18} + 139591494y^{20} + 293155569y^{22} \\
              & + 343995552y^{24} + 204720453y^{26} + 45717291y^{28} + 3433590y^{30}.
\end{aligned}
\end{equation}

    \item A weight enumerator for a putative $[[31,1,9]]$ code with $\epsilon_{\rm out} \approx  \frac{245939 \epsilon ^7}{59} +O(\epsilon^8)$, and threshold $\epsilon_*=0.209924$:
\begin{equation}
\begin{aligned}
    A(1, y) = & \, 1 + 573y^4 + 4548y^6 + 10872y^8 + 57435y^{10} + 1568526y^{12} \\
              & + 2104353y^{14} + 14274369y^{16} + 4369902y^{18} + 141384465y^{20} \\
              & + 311881446y^{22} + 382394646y^{24} + 166952583y^{26} + 30493908y^{28} \\
              & + 18244197y^{30}.
\end{aligned}
\end{equation}
\end{itemize}

\bibliographystyle{ssg}
\bibliography{qudit2024}

\end{document}